\documentclass[letterpaper,final,10pt,journal,twoside]{IEEEtran}

\usepackage[dvipdfmx]{graphicx} %
\usepackage{url}
\usepackage[cmex10]{amsmath}
\usepackage{amssymb}  
\usepackage{amsthm}
\usepackage{amsfonts}%
\usepackage[mathscr]{eucal}
\usepackage{multirow}
\usepackage{algorithm}
\usepackage{algorithmic}

\usepackage[caption = false]{subfig}

\usepackage{color}

\usepackage{cite}
\usepackage{array}

\usepackage{enumitem}

\usepackage{hyperref}

\newtheorem{thm}{Theorem}

\theoremstyle{definition}
\newtheorem{defn}[thm]{Definition}
\newtheorem{assum}[thm]{Assumption}
\newtheorem{prop}[thm]{Proposition}
\newtheorem{lem}[thm]{Lemma}
\newtheorem{coroll}[thm]{Corollary}

\theoremstyle{remark}
\newtheorem{rem}[thm]{Remark}

\newcommand{\trasp}{\ensuremath{^{\intercal}}}
\newcommand{\bfp}{\ensuremath{\pi}}

\newcommand{\bfx}{\ensuremath{\mathbf{x}}}
\newcommand{\bfu}{\ensuremath{\mathbf{u}}}
\newcommand{\bfw}{\ensuremath{\mathbf{w}}}

\newcommand{\bfn}{\ensuremath{\mathbf{n}}}
\newcommand{\bfq}{\ensuremath{\mathbf{q}}}

\newcommand{\bfA}{\ensuremath{\mathbf{A}}}
\newcommand{\bfB}{\ensuremath{\mathbf{B}}}
\newcommand{\bfQ}{\ensuremath{\mathbf{Q}}}
\newcommand{\bfR}{\ensuremath{\mathbf{R}}}

\newcommand{\setS}{\ensuremath{\mathcal{S}}}
\newcommand{\setT}{\ensuremath{\mathcal{T}}}
\newcommand{\setP}{\ensuremath{\mathcal{P}}}
\newcommand{\coalstr}{\ensuremath{\mathscr{P}}}

\newcommand{\setX}{\ensuremath{\mathcal{X}}}

\newcommand{\setE}{\ensuremath{\mathcal{E}}}

\newcommand{\setU}{\ensuremath{\mathcal{U}}}
\newcommand{\setV}{\ensuremath{\Psi}}

\newcommand{\coal}{\ensuremath{\mathcal{C}}}
\newcommand{\setN}{\ensuremath{\mathcal{N}}}
\newcommand{\setM}{\ensuremath{\mathcal{M}}}

\newcommand{\setO}{\ensuremath{\mathcal{O}}}

\newcommand{\mmg}{\ensuremath{\mathcal{G}}}
\newcommand{\globlyap}{\ensuremath{\mathscr{V}}}
\newcommand{\unione}{\scalebox{0.9}{$\cup$}}
\newcommand{\unionped}{\scalebox{0.6}{$\cup$}}

\newcommand{\vmerger}{\ensuremath{v(\setP_1\cup\setP_2)}}

\newcommand{\notimplies}{%
  \mathrel{{\ooalign{\hidewidth$\not\phantom{=}$\hidewidth\cr$\implies$}}}}
\newlength{\altfigur}
\newcounter{subeqn} %
\begin{document}
	\title{Coalitional control for self-organizing agents}
	\author{Filiberto Fele$^{1}$, Ezequiel Debada$^{2}$, Jos\'e M.~Maestre$^{3}$,~\IEEEmembership{Member,~IEEE} and Eduardo F.~Camacho$^{3}$,~\IEEEmembership{Fellow,~IEEE}
		\thanks{The final version of this manuscript can be found in~\cite{8253905}.%
		}%
		\thanks{Financial support by the FP7-ICT project DYMASOS (ref. 611281), Spanish MINECO project DPI2016-78338-R and Junta de Andaluc\'ia project P11-TEP-8129, is gratefully acknowledged.}
		\thanks{$^1$ F. Fele is with the Department of Engineering Science, University of Oxford, OX1 3PJ, UK
			{\tt\small filiberto.fele@eng.ox.ac.uk}}%
		\thanks{$^2$ E. Debada is with the React Group, EPFL, Lausanne, Switzerland
			}%
		\thanks{$^3$ The authors are with Departamento de Ingenier\'{i}a de Sistemas y Autom\'{a}tica, ETSI Universidad de Sevilla, 41092 Seville, Spain
			}%
	}
\maketitle
\begin{abstract}
Coalitional control is concerned with the management of multi-agent systems where cooperation cannot be taken for granted (due to, e.g., market competition, logistics). 
This paper proposes a model predictive control (MPC) framework aimed at large-scale dynamically-coupled systems whose individual components, possessing a limited model of the system, are controlled independently, pursuing possibly competing objectives. 
The emergence of cooperating clusters of controllers is contemplated through an autonomous negotiation protocol, based on the characterization as a coalitional game of the benefit derived by a broader feedback and the alignment of the individual objectives.
Specific mechanisms for the cooperative benefit redistribution that relax the cognitive requirements of the game are employed to compensate for possible local cost increases due to cooperation. As a result, the structure of the overall MPC feedback can be adapted online to the degree of interaction between different parts of the system, while satisfying the individual interests of the agents.
A wide-area control application for the power grid with the objective of minimizing frequency deviations and undesired inter-area power transfers is used as study case.
\end{abstract}

\section{Introduction}
Major challenges in control are in dealing with the increasing heterogeneity of networked systems---possibly characterized by decentralized management,  autonomy of the parts and dynamic structural reconfiguration capabilities~\cite{EngellEtAl2015}.
In such setting, selfish interests may assume a dominant role, significantly constraining the management of the system and their performance. This issue is especially evident in public infrastructures, often co-owned by independent entities, and whose management requires a tradeoff among sectors in direct competition~\cite{NEGENBORN10BOOK,MalandrinoEtAl2015b}.\par
Several works in the distributed control literature have studied the performance and stability issues for different modalities of participation of the control agents in the achievement of the global objective~\cite{RawlingsStewart2008,SCA09JPC}. As the systems become more complex and articulated, control architectures featuring flexible cooperation patterns have been recently proposed. For example, the notion of \emph{cooperating sets} of controllers appears in~\cite{TroddenRichards2009, TroddenRichards2013}. Within these sets, individual strategies are optimized considering what others may be able to achieve, thus indirectly promoting cooperation. The composition of the sets is updated according to a graph representing the active coupling constraints. In~\cite{NUNEZ13ICICAS} the hierarchy of the agents is adapted to different operational conditions by rearranging the order followed in the optimization of the control actions. The work of~\cite{FeleJPC} investigates the design of a hierarchical model predictive control (MPC) scheme for interconnected systems, where the sparsity pattern of the overall MPC feedback is dynamically adjusted to optimize the data link usage.\par
Methods for the analysis of the relevance of the agents and the communication paths involved in the distributed control of complex interconnected systems have been recently studied by~\cite{MaestreIshii2016_ACC,SummersLygeros2014,LiuEtAl2011}. The structural information provided by these methods allows the efficient allocation of the control resources, promoting sparsity in order to minimize computational and communicational requirements~\cite{WuJovanovic2017,GusrialdiHirche2010CDC}. A step further is the \emph{online} identification of the optimal control structure: besides accommodating the controller requirements in real-time~\cite{SchuhLunze2015ECC}, such flexibility grants the possibility of reconfiguring the system for improving robustness or fault-tolerance~\cite{GrossJilgStursberg2013_ECC}, or even for featuring plug-and-play capabilities~\cite{RiversoEtAl2013}.\par
The majority of the mentioned works addresses the achievement of a unique goal common to all the agents. Here we consider instead control agents that focus on some local (economically valuable) objectives. In such a scenario, a natural solution for steering individual interests towards the global welfare is the employ of incentive mechanisms from game theory. Several proposals are available so far for traffic or demand reshaping~\cite{DePaolaEtAl2017TranSmartGrid,Grammatico2017Aggreg,PfrommerEtAl2014}, and in competing markets like electric vehicles (EVs) recharge~\cite{GentileEtAl2017NashWardrop,MalandrinoEtAl2015a, YuanEtAl2016}.\par
When \emph{active} cooperation is a possibility, then \emph{individual rationality} of the agents needs to be taken into account, as the cooperation will be strictly associated with the expected share of the collective benefit. In this context, game theoretic tools for the redistribution of the value of cooperation are fundamental. 
A decentralized algorithm for benefit redistribution among cooperating agents is proposed in~\cite{NedicBauso2013TAC}. The bargaining protocol is run on a time-varying communication graph, and the resulting allocation is proven to converge to a stable one, that is, satisfying all players. The work of~\cite{ValenciaThesis2012} provides a cooperative MPC formulation where cooperation is subject to bargaining. The satisfaction of a minimum individual performance is imposed by a \emph{disagreement point}, defined as the threshold of maximum allowed loss of performance in case of cooperation. A cooperative distributed MPC scheme prioritizing local objectives is presented in~\cite{LopesDeLimaEtAl2015}: following situational altruism criteria, local objectives are dynamically adjusted to fulfill minimum local cost requirements. In~\cite{LianEtAl2017_JSAC}, the impact of sparsity constraints on the LQR global feedback law is analyzed from a communication cost point of view. In particular, as the sparsity constraint is relaxed to enhance system performance, the reallocation of the communication costs over the cooperating agents is studied.
In~\cite{RamosEtAl2013}, self-organizing coalitions among EVs are considered as a means to enhance the predictability of the vehicle-to-grid offer, by presenting a wider energetic portfolio to the grid operator. 
Analogously, the work of~\cite{BaeyensEtAl2013} studies the formation of coalitions among wind energy producers with the objective of reducing the output variability in the aggregate offer, and so improve their expected profit. The authors of~\cite{MalandrinoEtAl2015b} investigate how the equilibrium can be reached in an EV recharging market whose actors are (coalitions of) charging stations and EV users.\par
The coalitional distributed MPC architecture described in this paper is based on analogous game-theoretical grounds. In particular, we consider coalitions as a means for control agents to reduce the effect of the externalities represented by the (otherwise unknown) dynamical coupling imposed on one another.
Even if there are clear incentives---from a cooperative distributed control standpoint---for all agents to come together in the interests of minimizing such externalities, we consider here possible inefficient situations, arising from structural limitations or informational constraints, that may lead to the formation of intermediate coalitions~\cite{ray2007game}.
More specifically, we consider that a set of global MPC control laws is associated with the possible cooperation structures of the control agents, and propose a framework allowing to study the resulting global switching behaviour.
The global cooperation structure emerges as the outcome of the \emph{autonomous} coalition formation between the agents, through a pairwise bargaining procedure where costs of cooperation are taken into account.
The main element of the bargaining is the \emph{online} redistribution of the value of a coalition. 
In particular, it is shown that convergence to a stable allocation of the coalitional benefit can be obtained without imposing a heavy cognitive demand on the agents, thus maintaining compatibility with the restricted communication characterizing the considered scenario. 
This is achieved on the basis of an iterative mechanism guaranteeing coalition-wise stability, provided that the \emph{core} of the associated transferable-utility (TU) game is nonempty~\cite{Stearns1968,Cesco98aconvergent,SANDHOLM1999}. The contributions of the paper include the formalization of design conditions concerning closed-loop stability and nonemptiness of the core. The analysis shows how, when global model information is locally unavailable, cooperation costs play a major role on the outcome of the coalition formation, and that these can be used as a mechanism to link coalition formation with desired closed-loop properties.
 Finally, the effectiveness of the proposed coalitional control framework is demonstrated on a wide-area control (WAC) application in power grids, with the objective of minimizing frequency deviations and undesired inter-area power transfers.\par
The document is organized as follows: Section~\ref{sec_probstat} introduces the model of the system and of the communication infrastructure; the controller and the ingredients employed for coalition formation are formalized in Section~\ref{sec_coal_contr}; the utility transfer scheme and the conditions for nonemptiness of the core are discussed in Section~\ref{sec_coalstab}; Section~\ref{sec_bargaining} presents the algorithms for coalition formation/splitting and the derivation of individual cost allocation. Section~\ref{sec_example} illustrates numerical results on a power grid application.\par
\emph{Notation}: All vectors are intended as column vectors, unless differently specified. Given a set $\setN = \{1,\ldots,n\}$, $(x_i)_{i\in\setN}$ denotes the column vector $(x_1, x_2, \ldots, x_n)$ obtained by stacking all (column) vectors $x_i$, for all $i\in\setN$. State and input vectors relative to coalitions are notated in bold: thus $\bfx_i$ is the state vector of coalition $i$, whereas $x_j$ denotes the state of subsystem $j$.
 $x(t|k)$ denotes the value of $x(k+t)$ estimated at time $k$. $\mathscr{L}$ is the set of functions $\varphi: \mathbb{N}\rightarrow [0,\infty)$, such that $\varphi(\cdot)$ is decreasing and $\lim_{t\rightarrow\infty}\varphi(t)=0$.
\section{Problem statement}
\label{sec_probstat}
\subsection{System description}
\label{sec_sysmodel}
Consider a system that can be described as a collection $\setN = \{1,\ldots,|\setN|\}$ of coupled linear processes, each governed by a local control agent, and modeled by the following discrete-time state-space equations:
\begin{subequations}\label{eq_ind_ss}
\begin{align}
x_i(k+1) &= A_{ii} x_i(k)+ B_{ii}u_i(k) + w_i(k),\label{eq_ss_a}\\
w_i(k) &= \sum_{j\in \setM_i} A_{ij}x_j(k) + B_{ij}u_j(k),\label{eq_ind_disturb}
\end{align}
\end{subequations}
where $x_i\in\mathbb{R}^{n_i}$ and $u_{i}\in\mathbb{R}^{q_i}$ are respectively the state and local control input vectors of subsystem $i\in\setN$, constrained in the sets $\setX_i$ and $\setU_i$ respectively.\footnote{Without loss of generality and for notational convenience, we assume in the remainder that $n_i=n_j$ and $q_i=q_j$ for all $i,j\in\setN$.}
Matrices $A_{ii},B_{ii}$ are properly sized state-transition matrices relative to the local states and inputs. Similarly, $A_{ij},B_{ij}$ are the matrices describing the coupling $w_i\in\mathbb{R}^{n_i}$ with states and inputs of \emph{neighbor} subsystems. The neighborhood set is defined as
$\setM_i = \left\{j\in \setN\setminus\{i\}:\, A_{ij} \neq \mathbf{0} \vee B_{ij} \neq \mathbf{0}\right\}$.
Models analogous to~\eqref{eq_ind_ss} have been employed for the control of large-scale systems such as drinking water networks composed of
interconnected water tanks~\cite{Ocampo2012HierarcMPC}, irrigation canals~\cite{Schuurmans1997,FeleJPC,Negenborn:08c}, supply chains~\cite{JMM09CDCb,JMM11JPC}, traffic networks~\cite{CamponogaraOliveira2009} and power grids~\cite{VenkatEtAl2006CDC}.\par
Denoting the global state as $x = (x_i)_{i\in\setN}\in\mathbb{R}^n$ and the global input as $u = (u_i)_{i\in\setN}\in\mathbb{R}^q$, the state evolution of the whole system of systems is governed by the following equation
\begin{equation}
x(k+1) = A x(k) + B u(k),
\label{eq_sos_ss}
\end{equation}
where $A=[A_{ij}]_{i,j\in\setN}\in\mathbb{R}^{n\times n}$ and $B=[B_{ij}]_{i,j\in\setN}\in\mathbb{R}^{n\times q}$, $n=\sum_{i\in\setN}n_i$, $q=\sum_{i\in\setN}q_i$. We designate the global system constraints as $\setX = \prod_{i\in\setN}\setX_i$ and $\setU = \prod_{i\in\setN}\setU_i$.\par
\subsection{Exchange of information}
\label{sec_infoexchange}
Control agents can communicate through a network infrastructure schematized by the undirected graph $\mmg(k) = (\setN,\setE(k))$, where $\setE(\cdot)\subseteq \setN \times \setN$. We consider a time variant set of links $\setE(\cdot)$ reflecting the possibility of establishing/disrupting communication links at some given time steps. 
In particular, let $\mathcal{T}_{\coal}\subseteq\mathbb{N}$. For any two consecutive elements $k',k''\in\mathcal{T}_{\coal}$, with $k'<k''$, we have $\setE(k)=\setE(k')$ for all $k\in\{k',\ldots,k''-1\}$. 
Any communication link $(i,j)\in\setE(k)$ defines the mutual availability of state (and input) feedback information between agents $i,j\in\setN$. 
Thus, the graph $\mmg(k)$ delineates a partition $\coalstr(\mmg(k)) = \{\coal_1,\ldots,\coal_{n_c}\}$ of the set of controllers into $n_c\in\left[1,|\setN|\right]$ connected components, referred to as (non-overlapping) coalitions, such that $\coal_i \subseteq \setN,\, \coal_i\cap\coal_j = \varnothing, \, \forall i,j\in\{1,\ldots,n_c\},i\neq j, \, \text{and } \bigcup_{i=1}^{n_c} \coal_i =\setN$~\cite{RahwanEtAl2012}.\par
In other words, $\coalstr(\mmg(k))$ reflects the (sparse) global control feedback structure. 
 Each coalition $\coal_r\in\coalstr(\mmg(k))$ can be considered as a unique system, where the dynamics~\eqref{eq_ss_a} of all subsystems involved are aggregated as
\begin{equation}
\bfx_r(k+1) = \bfA_{rr} \bfx_r(k) + \bfB_{rr} \bfu_r(k) + \bfw_r(k), 
\label{eq_coal_ss}
\end{equation}
where $\bfx_r = (x_i)_{i\in\coal_r}\in\mathbb{R}^{\bfn_i}$ is the aggregate state vector, and $\bfA_{rr} = [A_{ij}]_{i,j\in\coal_r}$ the relative state transition matrix, describing the state coupling between members of the same coalition. The components of $\bfu_r = (u_i)_{i\in\coal_r}\in\mathbb{R}^{\bfq_i}$ are the local control inputs of the subsystems in $\coal_r$, and $\bfB_{rr} = [B_{ij}]_{i,j\in\coal_r}$ is the associated coalitional input matrix.
Finally, the vector $\bfw_r=(w_i^{(r)})_{i\in\coal_r}$ gathers the disturbance due to the coupling with subsystems external to the coalition. For each $i\in\coal_r$ we have
\begin{equation}
w_i^{(r)} = \sum_{j\in\setM_i\setminus\coal_r} A_{ij}x_j(k)+ B_{ij}u_j(k),
\label{eq_disturb_coal2}
\end{equation}
and $w_i^{(r)} = \mathbf{0}$ if $\setM_i\setminus\coal_r = \varnothing$. Note that the definition of $w_i^{(r)}$ is equivalent to~\eqref{eq_ind_disturb} except the sum is restricted to $\setM_i\setminus\coal_r$. Thus, from the coalition standpoint, the modeling uncertainty comes from subsystems $j\in\left(\bigcup_{i\in\coal_r}\setM_i\right)\setminus\coal_r$.

\section{Coalitional control}
\label{sec_coal_contr}
Cooperation between local control agents translates into better performances~\cite{RawlingsStewart2008}. This comes however at the expense of higher communication and computation requirements~\cite{MaestreEtAl2015OCAM}.
Indeed, the effort required for the coordination increases with the number of agents involved in a coalition. 
Costs incurred for cooperation can be taken into account by means of ad-hoc indices related to, e.g., the size of the coalition, the distance between its members~\cite{FeleEtAl2016CSM}, the number of data links needed to establish communication between them~\cite{MAESTRE13OCAM,FeleJPC}. 
The design of a networked controller architecture can be formulated as a trade-off between control performance and savings on the coordination costs~\cite{LianEtAl2017_JSAC,WuJovanovic2017,JilgStur2013_IFAC,ALESSIO2011}.\par
In this paper we propose a game theoretical framework for the dynamic establishment of cooperation in the control of a multi-agent system. 
The presence of an omniscient supervisor is not assumed here: the cooperation between any two parties is established \emph{autonomously}, as the outcome of a pairwise bargaining between the coalitions in $\coalstr(\mmg(\cdot))$. The object of the bargaining is the \emph{reallocation} of the benefit derived from coordination.
Thus, the overall cooperation structure dynamically evolves following a trade-off between increased performance and costs incurred for cooperation.\par
From now on, the parties involved in a bargaining over the formation of a (bigger) coalition will be designated as \emph{players} 1 and 2;
for notational convenience, the index `$1\unione 2$' will refer to their merger. Note that the term player may refer to either a single control agent or a group of agents that, as a consequence of their participation in the same coalition, act as a single entity. Formally, these agents are identified by the sets $\setP_1,\setP_2\in\coalstr(\mmg(k))$.
Before defining the criterion for the coalition formation bargaining, we discuss the performance improvement offered by cooperative control---viewed as coalitional benefit---and highlight the issues of the absence of benefit redistribution from the (economic) standpoint of the individual agents.
\subsection{Control objective}
\label{sec_controbj}
We consider control agents $j\in\setN$ implementing an optimal control policy aimed at minimizing a local (quadratic) stage cost $\ell_j: \mathbb{R}^{n_j}\times \mathbb{R}^{q_j}\rightarrow \mathbb{R}$, over an horizon $N_p$. In particular, we assume that this optimal control policy is derived through a model predictive control (MPC) approach~\cite{EFC1}. We will refer to $\ell_j(x_j,u_j)$ as the \emph{selfish} objective. It is worth to point out that the selfish objective is implicitly a function of other systems' states, through the coupling in~\eqref{eq_ind_ss}. The impact of this coupling on the local cost is uncertain unless cooperation is introduced.
 Within each $\coal_i\in\coalstr$ the coalitional stage cost is defined as
 $\Lambda_i(\bfx_i,\bfu_i) : \mathbb{R}^{\bfn_i}\times \mathbb{R}^{\bfq_i}\rightarrow \mathbb{R}$, with $\bfn_i = \sum_{j\in\coal_i}{n_j}$, $\bfq_i = \sum_{j\in\coal_i}{q_j}$.
Built upon the selfish objectives, the coalitional objective allows to improve on them by exploiting the shared feedback information available at coalition level and explicitly include the coupling variables in its formulation.
Following an MPC approach, at time $k$ a control input for $\coal_i$ is derived from the solution of the optimization problem~\cite{EFC1,RawlingsLIB09} 
\begin{subequations}
\label{eq_coal_MPC}
\begin{align}
\bfu_i^{\ast} = \arg\min_{\bfu_i} & = \sum_{t=0}^{N_p-1} \Lambda_i(\bfx_i(t|k),\bfu_i(t|k)) + V_i^{\mathrm{f}}(\bfx(N_p|k)) \label{eq_coal_MPC_cost} \\
\mathrm{s.t.}& \nonumber\\
\bfx_i(t+1|k) & = \bfA_{ii}\bfx_i(t|k) + \bfB_{ii}\bfu_i(t|k),\label{eq_mod_MPC}\\
\bfx_i(t|k) &\in \prod_{j\in\coal_i}\setX_j,\,t = 0,\ldots,N_p-1,\label{eq_restrx_MPC}\\
\bfu_i(t|k)& \in \prod_{j\in\coal_i}\setU_j ,\,t = 0,\ldots,N_p-1,\label{eq_restru_MPC}\\
\bfx_i(N_p|k) &\in \Omega_i,\label{eq_restrfin_MPC}\\
\bfx_i(0|k) &= \bfx_i(k),\label{eq_init_constr}
\end{align}
\end{subequations}
where~\eqref{eq_mod_MPC} is the prediction model for the evaluation of the cost function~\eqref{eq_coal_MPC_cost} over the horizon of length $N_p$; 
the second term in~\eqref{eq_coal_MPC_cost}, $V_i^{\mathrm{f}}(N_p|k)$, denotes the terminal cost. 
$\Omega_i\subseteq \prod_{j\in\coal_i}\setX_j$ is a terminal set constraint~\cite{MayneEtAl2000Stability}.
The first element of $\bfu_i^{\ast} \triangleq (\bfu_i^{\ast}(0|k),\, \ldots \, \bfu_i^{\ast}(N_p-1|k))$ is applied at time $k$ to the subsystems involved in the coalition, i.e., $\bfu_i(k)\triangleq \bfu_i^{\ast}(0|k) = (u_j^{\ast}(0|k))_{j\in\coal_i}$, and~\eqref{eq_coal_MPC} is solved again at subsequent time instants in a receding horizon fashion.\par
Problem~\eqref{eq_coal_MPC} is solved \emph{independently} by each coalition $\coal_i\in\mathscr{P}(\setN,\mmg(k))$. The computation of~\eqref{eq_coal_MPC} can be performed by a \emph{coalition leader}, or distributed across the members of the coalition. Several algorithms are available for the distributed solution of convex MPC problems, see, e.g.,~\cite{MAE14DME}.\par
In case of singleton coalition, i.e., $\coal_i = \{i\}$, $i\in\setN$,~\eqref{eq_coal_MPC} corresponds to the selfish optimization control problem.
When all the agents are pursuing their own selfish objective through a local state feeedback, a decentralized \emph{noncooperative} feedback law emerges globally. In contrast, when the \emph{grand coalition}  is formed, the solution of~\eqref{eq_coal_MPC} coincides with a centralized MPC feedback law. 
In all other cases, a \emph{semi-cooperative} global feedback law is implemented.
\begin{rem}
Although in absence of cooperation costs the centralized MPC feedback law results in the~\emph{social} optimum, in the setting considered here the grand coalition is not necessarily the most efficient cooperation structure.
\end{rem}
\begin{rem}
Reflecting the state feedback structure imposed by $\coalstr(\mmg(\cdot))$, $\bfw_i$ is absent in the prediction model. Although the performance of the MPC control law might be enhanced by including an estimated external coupling term in~\eqref{eq_mod_MPC}, its derivation is in general application-oriented and out of the scope of this paper.
\end{rem}
\begin{assum}[Weak coupling]
\label{assum_weak}
All subsystems are input-to-state stable (ISS) when controlled with the MPC feedback law $\kappa_{i}: \prod_{j\in\coal_i}\setX_j \rightarrow \prod_{j\in\coal_i}\setU_j$ derived from~\eqref{eq_coal_MPC}, treating~\eqref{eq_ind_disturb} as an unknown disturbance. Moreover, the small-gain condition for the interconnected systems holds for the global control law $\kappa_{\coalstr}: \mathcal{X} \rightarrow \mathcal{U}$, $\kappa_{\coalstr}\triangleq (\kappa_{i})_{i\in\coalstr}$ associated with each possible $\coalstr(\mmg(\cdot))$, and~\eqref{eq_sos_ss} is ISS~\cite{GeiselhartLazarWirth2015}.
\end{assum}
\begin{assum}[Dwell time]
\label{assum_dwelltime}
There exist a set $\setT_{\coal}\subseteq\mathbb{N}$ defining the switching instants,
and $\bar{\tau}_D > 0$ such that for every two consecutive elements $k',k''\in\mathcal{T}_{\coal}$ it holds $k''-k'\geq \bar{\tau}_D$, and for which the system~\eqref{eq_sos_ss} in closed loop with the set of switched control laws $\{\kappa_{\coalstr}: \mathbb{X} \rightarrow \mathbb{U}\}$, associated with each possible $\coalstr(\mmg(\cdot))$, is ISS~\cite{VuEtAl2007ISSdwelltime,HespanhaMorse1999dwelltime}.
\end{assum}
\begin{rem}
Several characteristics of model predictive control make it the ideal choice in this setting, e.g., clear definition of the performance objectives, direct consideration of input and state constraints. Nonetheless, the essential feature that facilitates the coalitional framework is the receding-horizon evaluation of the controller performance---intrinsic to MPC control. This will be clear by the next section.
\end{rem}

\subsection{Evaluation of coalitional benefit}
\label{sec_COO}
In the following we will employ an index expressing the control performance and cooperation costs associated to a given coalition.
\begin{defn}[Coalition value]
We define the \emph{value} of coalition $\coal_i\subseteq\setN$ as the function $v: 2^{\setN}\rightarrow \mathbb{R}$,
\begin{equation}
v(\coal_{i}) = \sum\limits_{t=0}^{N_p-1} \Lambda_{i}(\bfx_i(t|k),\bfu_i(t|k)) + \chi_{i}(\mmg_i(k)),
\label{eq_index_COO}
\end{equation}
where the cooperation cost $\chi_i(\cdot): (\setN,\setE)\rightarrow [0,\infty)$ depends on the subgraph describing the connections between the members of $\coal_i$. We assume here that $\chi_i(\cdot)$ is monotone increasing in the number of nodes in the graph.
\end{defn}
Given a pair of players $\setP_1,\setP_2\in\coalstr$,~\eqref{eq_index_COO} is evaluated for the two players separately (\emph{unilateral} strategies) and for their merger (\emph{coalitional} strategy). 
\begin{rem}
The evaluation of~\eqref{eq_index_COO} requires the \emph{mutual interaction} model between (the subsystems in) $\setP_1$ and $\setP_2$ in the solution of~\eqref{eq_coal_MPC}. This model is assumed available during the bargaining process.
\end{rem}
\subsubsection{Evaluation of the merger}
$\Lambda_{i}(\cdot,\cdot)$ is evaluated with the input sequences $\bfu_{1\unionped 2}^{\ast}$ and the associated predicted state trajectories $\bfx_{1\unionped 2}^{\ast}$ obtained as the solution of~\eqref{eq_coal_MPC} relative to the coalition $\setP_{1}\cup\setP_{2}$. 
$\chi_{1\unionped 2}(\mmg_{1\unionped 2}(k))$ is evaluated over the subgraph describing all the connections $(i,j)\in\setE_{1\unionped 2}(k)$ between every pair of agents $i,j\in(\setP_1 \cup \setP_2)$.
We refer to the jointly optimized input sequence $\bfu_{1\unionped 2}^{\ast}$ as \emph{coalitional} strategy.
\subsubsection{Evaluation of unilateral strategies}
If the players are dynamically coupled, their optimal trajectories will be interdependent. Therefore, 
a consistent evaluation of unilateral strategies can only be performed if some knowledge about the input and state sequences applied by the other player is available. 
Unilateral strategies can be derived over an iterative procedure, as follows:
\textit{(i)} set $\tilde{\bfu}_{j} \triangleq {\bfu}_{j}^{(l-1),\ast}$, i.e., the optimal control sequence computed at iteration $l-1$ by player $j\neq i$, and $\tilde{\bfx}_{j} \triangleq {\bfx}_{j}^{(l-1),\ast}$, its associated state trajectory;
\textit{(ii)} solve~\eqref{eq_coal_MPC} for $i,j=\{1,2\}$, $j\neq i$, replacing~\eqref{eq_mod_MPC} with 
\begin{multline}\label{eq_mpc_unilateral}
\bfx_{i}^{(l)}(t+1|k)  = \bfA_{ii} \bfx_{i}^{(l)}(t|k) + \bfB_{ii} \bfu_{i}^{(l)}(t|k)\\
	+ \bfA_{ij} \tilde{\bfx}_{j}(t|k) + \bfB_{ij} \tilde{\bfu}_{j}(t|k). 
\end{multline}
The tails of the optimal sequences computed at time $k-1$ can be used as initial trajectories, i.e., ${\bfu}_{j}^{(0),\ast}(t|k) \triangleq \bfu_{j}^{\ast}(t+1|k-1)$, provided they are feasible.
Finally, $\chi_{i}(\mmg_i(k))$ is evaluated over the subgraph describing all the connections $(j,j')\in\setE_i(k)$ between every pair of agents $j,j'\in\setP_i$.
\subsection{Individual rationality}
\label{sec_CIR}
The premise here is that agents are rational: they accept to cooperate only if the redistribution of the coalition benefit constitutes an improvement upon the outcome of the unilateral strategy. A necessary condition for this is that the benefit outperforms the aggregate outcome of unilateral strategies, i.e., 
\begin{equation}
v(\setP_1\cup\setP_2) \leq v(\setP_1)+ v(\setP_2),
\label{eq_COO_OK}
\end{equation}
with $v(\cdot)$ defined in~\eqref{eq_index_COO}.
Let the cost incurred by player $i\in\{1,2\}$ under the coalitional strategy $\bfu_{1\unionped 2}^{\ast}$ be
\begin{multline}\label{eq_portion_player}
\vmerger|_{(i)} \triangleq \sum\limits_{t=0}^{N_p-1} \Lambda_{i}(\bfx^{\ast}_{1\unionped 2}(t|k),\bfu^{\ast}_{1\unionped 2}(t|k))\\ + \chi_{1\unionped 2}(\mmg_{1\unionped 2}(k))|_{(i)},
\end{multline}
where, with an abuse of notation, $\Lambda_{i}(\bfx^{\ast}_{1\unionped 2}(\cdot),\bfu^{\ast}_{1\unionped 2}(\cdot))$ means that the influence of the coupling of player $j\in\{1,2\}\setminus\{i\}$ on the cost of player $i$ is taken into account for the computation of $\Lambda_i$.
 Note that $\vmerger|_{(1)} + \vmerger|_{(2)} = \vmerger$, and the value of $\chi_{1\unionped 2}(\mmg_{1\unionped 2}(k))|_{(j)}$ is a proper (predefined) allocation of the cooperation costs. 
We are now ready to formally state the leading thread of this work. It can be easily verified by example, so we present it without proof.
\begin{prop}\label{prop_ind_ration}
Condition~\eqref{eq_COO_OK} does not imply lower incurred costs to both players, i.e.,
\begin{multline}
\vmerger \leq v(\setP_{1}) + v(\setP_{2})\\ \notimplies \vmerger|_{(j)} \leq v(\setP_{j}),\; \forall j\in\{1,2\},
\label{eq_ind_ration1}
\end{multline}
where $v(\cdot)$ is defined in~\eqref{eq_index_COO}, and $v(\cdot)|_{(j)}$ is defined in~\eqref{eq_portion_player}.
\end{prop}
We refer to the RHS in~\eqref{eq_ind_ration1} as the \emph{individual rationality} requirement.
In other words, the merger forms if and only if $\vmerger|_{(j)} \leq v(\setP_{j})$ is fulfilled for both players $j\in\{1,2\}$.\par
The same argument leading to~\eqref{eq_ind_ration1} can be extended to the individual members of any coalition. Even if cooperation allows to decrease the aggregate cost, it can indeed be unfavorable for some agents from the point of view of the locally incurred costs. Thus, there is not a straightforward relationship between cooperation and \emph{individual rationality}---unless some means of transferring the value between agents is provided. 
\subsection{Transferable utility}
\label{sec_TUalg}
\begin{assum}[Transferable utility (TU)]\label{assum_TU}
The coalitional performance~\eqref{eq_index_COO} is an economic index. A value equivalent to $v(\setP_1) + v(\setP_2) - \vmerger$, i.e., the \emph{surplus} of the merger, can be reallocated between the agents.
\end{assum}
Let $p_j^{(i)}\in\mathbb{R}$ designate the cost reallocated to agent $j\in\coal_i\subseteq\setN$, and $p^{(i)}=(p_j^{(i)})_{j\in\coal_i}$ the vector of allocations to the members of $\coal_i$. 
\begin{defn}[Efficiency]
An allocation is \emph{efficient} w.r.t. the coalition value $v(\coal_i)$ if $\sum_{j\in\coal_i}{p_j^{(i)}} = v(\coal_i)$.
\end{defn}
We refer to $\bfp_{i} \triangleq \sum_{j\in\coal_i}{p_j^{(i)}}$ as the \emph{aggregate} cost allocated over coalition $\coal_i$.
\begin{lem}
Let \eqref{eq_COO_OK} and Assumption~\ref{assum_TU} hold for a given tuple $\{\coal_1,\coal_2,\coal_1\cup\coal_2\}$, where $\coal_1,\coal_2\in\coalstr(\mmg(k))$. Then there exists an efficient allocation of the merger cost, i.e., vectors $p^{(1)}\in\mathbb{R}^{|\coal_1|}$ and $p^{(2)}\in\mathbb{R}^{|\coal_2|}$ such that
\begin{multline*}
\sum_{j\in\coal_i}{p_j^{(i)}} \leq v(\coal_i),\text{ for }i\in\{1,2\},\\ \text{ and } \sum_{j\in\coal_1}{p_j^{(1)}}+\sum_{j\in\coal_2}{p_j^{(2)}} = v(\coal_1\cup\coal_2).
\end{multline*}
\end{lem}
\begin{proof}
A straightforward solution is the \emph{egalitarian} redistribution, where an \emph{equal} share of the merger surplus is assigned to each player, i.e.,
\begin{equation}
\bfp_i = v(\coal_i) - \frac{1}{2}\left(v(\coal_i) + v(\coal_j) - v(\coal_i\cup \coal_j)\right), 
\label{eq_Shapley}
\end{equation}
for $i\in\{1,2\}$ and $j\neq i$. Geometrically, the allocation $(\bfp_1,\bfp_2)$ corresponds to the midpoint of the line segment connecting $(v(\coal_1),v(\coal_1\cup\coal_2)-v(\coal_1))$ and $(v(\coal_1\cup\coal_2)-v(\coal_2),v(\coal_2))$ (and also coincides with the Shapley value formula for a two-player game)~\cite{Ferguson_GTBook}.
\end{proof}
So far, the discussion has been carried out without explicitly dealing with the case $|\setP_i|>1$. 
In Section~\ref{sec_coalstab} we address the redistribution of the aggregate cost $\bfp_i$ allocated to coalition $\coal_i$ over each one of its members.
\subsection{Closed-loop performance}
In this section we discuss the closed-loop performance of the proposed coalitional MPC control scheme. More specifically, we address the deviation between the predicted and the closed-loop control cost as a consequence of the formation of a coalition.\par
Consider $\coal_i\in\coalstr(\mmg(k))$, and assume that $\coalstr(\mmg(k))$ does not change in the interval $[k,k+N_p-1]$. 
Let
\begin{equation}
\varpi_i(k) \triangleq  \sum_{t=k}^{k+N_p-1} \Lambda_i(\bfx_i(k),\bfu_i(k)) -  \sum_{t=0}^{N_p-1} \Lambda_i(\bfx_i(t|k),\bfu_i(t|k))
\label{eq_diffpredreal}
\end{equation}
measure the deviation  between the predicted cost and the cost actually incurred in closed loop, i.e., when all the agents apply the optimal trajectories $\bfu_i(k) = \bfu_i^{\ast}(0|k)$ over $t = \{k,\ldots,k+N_p-1\}$ (since coupling from external subsystems is neglected in~\eqref{eq_mod_MPC}, we expect $\varpi_i(k) \neq 0$). We have seen in the previous section that the formation of a coalition is associated with an expected decrease in the control cost---derived through the jointly optimized control law---that (at least) compensates for the increase in the coordination effort. The aim of this section is to define the conditions under which the \emph{global} control cost does not increase upon the formation of a coalition.
\begin{assum}\label{assum_predimprove}
For any tuple $\{\coal_1,\coal_2,\coal_1\cup\coal_2\}$, with $\coal_1,\coal_2\in\coalstr$, we assume that $|\varpi_{1\unionped 2}(\cdot)|\leq |\varpi_{1}(\cdot)| + |\varpi_{2}(\cdot)|$.
\end{assum}
The above assumption implies that the dynamical effect of the rest of agents $\setN\setminus(\coal_1\cup\coal_2)$ on the subset $\coal_1\cup\coal_2$ does not change irrespective of how the agents in $\coal_1\cup\coal_2$ organize themselves into coalitions. 
 Observe that as a consequence of~\eqref{eq_disturb_coal2} and Assumption~\ref{assum_weak}, Assumption~\ref{assum_predimprove} is generally mild.\par
The following result, concerning the stability of the closed loop, follows from Assumption~\ref{assum_predimprove} and the input-to-state stability of the interconnected system.
\begin{thm}\label{lem_lyap} 
Consider $\coalstr = \{\coal_1,\ldots,\coal_{n_c}\}$, $\coalstr^+ = \coalstr\setminus\{\coal_1,\coal_2\}\cup \{\coal_1\cup\coal_2\}$, and let $\globlyap$ and $\globlyap^+$ be the associated closed-loop global costs over the interval $[k,k+N_p-1]$. Let $\Lambda_i(\cdot,\cdot)$ be a convex function, and let Assumption~\ref{assum_predimprove} hold. Then there exists a cooperation cost function $\chi_{i}(\cdot): (\setN,\setE)\rightarrow [0,\infty)$ for which $\globlyap^+ \leq \globlyap$.
\end{thm} 
\begin{proof}
Let $J_i = \sum_{t=0}^{N_p-1} \Lambda_i(\bfx_i^{\ast}(t|k),\bfu_i^{\ast}(t|k))$. Now notice that~\eqref{eq_COO_OK} must hold for $\coalstr^+$ to be a successor structure to $\coalstr$. Then, from~\eqref{eq_diffpredreal} we can write
\begin{subequations}
\begin{align} 
\globlyap^+ & \triangleq J_{1\unionped 2}  + \varpi_{1\unionped 2} + \sum_{r=3}^{n_c} (J_r + \varpi_r')\nonumber\\ 
&\qquad\qquad \leq J_{1\unionped 2}  + |\varpi_{1\unionped 2}| + \sum_{r=3}^{n_c} (J_r + \varpi_r') \\
& = v(\coal_1\cup\coal_2) - \chi_{1\unionped 2}(\coalstr^+) + |\varpi_{1\unionped 2}| + \sum_{r=3}^{n_c} (J_r + \varpi_r') \label{eq_proof11a}\\
& \leq v(\coal_1) + v(\coal_2) - \chi_{1\unionped 2}(\coalstr^+) + |\varpi_{1\unionped 2}| + \sum_{r=3}^{n_c} (J_r + \varpi_r') \label{eq_proof11b}\\
& \leq J_1 + J_2 + \chi_{1}(\coalstr) + \chi_{2}(\coalstr) - \chi_{1\unionped 2}(\coalstr^+) \nonumber\\
&\qquad\qquad\qquad + |\varpi_1| + |\varpi_2|
+ \sum_{r=3}^{n_c} (J_r  + \varpi_r') \label{eq_proof11c}\\ 
& \leq \globlyap + \chi_{1}(\coalstr) + \chi_{2}(\coalstr) - \chi_{1\unionped 2}(\coalstr^+)\nonumber\\
&\qquad\qquad +2(|\varpi_1| + |\varpi_2|) -  \sum_{r=3}^{n_c} \varpi_r + \sum_{r=3}^{n_c} \varpi_r',\label{eq_proof11d}
\end{align}
\end{subequations}
where the inequality~\eqref{eq_proof11a} follows from~\eqref{eq_index_COO}, \eqref{eq_proof11b} from~\eqref{eq_COO_OK}, \eqref{eq_proof11c} from Assumption~\ref{assum_predimprove}, \eqref{eq_proof11d} from the definition of $\globlyap$ and the upper bound to  
the worst case in which $\varpi_{1},\varpi_{2}<0$. Then
\begin{multline}
\globlyap^+ \leq \globlyap \Leftrightarrow \sum_{r=3}^{n_c} \varpi_r' \leq \sum_{r=3}^{n_c} \varpi_r \\
- 2(|\varpi_1| + |\varpi_2|) + \chi_{1\unionped 2}(\coalstr^+) - \chi_{1}(\coalstr) - \chi_{2}(\coalstr).
\label{eq_boundcostincr}
\end{multline}
The proof is concluded by observing that since $\setX$ and $\setU$ are compact, $\|\sum_{r=3}^{n_c} \varpi_r' - \sum_{r=3}^{n_c} \varpi_r + 2(|\varpi_1| + |\varpi_2|)\|$ is bounded, and the desired properties for $\chi_{i}(\cdot)$, i.e., monotone increasing in the coalition size, are fulfilled.
\end{proof}
\begin{rem}
In practice, the cases in which either $\varpi_{1}<0$ or $\varpi_{2}<0$ (or both) and $\varpi_{1\unionped 2}>\varpi_{1}+\varpi_{2}$ can be considered singular, in the sense that they are associated with mutual synergetic actions between the two players. In these cases, in presence of nonnegligible cooperation costs the players will likely be better off with unilateral strategies than with coalitional ones, and the bound in~\eqref{eq_boundcostincr} can be reduced to $\sum_{r=3}^{n_c} \varpi_r' \leq \sum_{r=3}^{n_c} \varpi_r + \chi_{1\unionped 2}(\coalstr^+) - \chi_{1}(\coalstr) + \chi_{2}(\coalstr)$.
\end{rem}

\section{Coalitional stability}
\label{sec_coalstab}
Consider again the vector $p^{(i)}=(p_j^{(i)})_{j\in\coal_i}$ of allocations to the members of $\coal_i$. We have seen so far that rational agents $j\in\setN$ will choose an allocation $p_j^{(a)}$ (associated to a coalition $\coal_a\subseteq\setN$) over $p_j^{(b)}$ (associated to another coalition $\coal_b\subseteq\setN$) if $p_j^{(a)} < p_j^{(b)}$.\par
Given a coalition $\coal\in\coalstr(\mmg(k))$, we seek allocations of $v(\coal)$ such that no agent $j\in\coal$ has incentive to leave the coalition. 
For this we consider the cooperative TU game $\Gamma_{\setN}(k)=\langle \setN,v \rangle$, and restrict our attention on the subgame $\Gamma_{\coal}(k)=\langle \coal,v \rangle$, where the \emph{characteristic} function $v(\cdot)$ is defined in~\eqref{eq_index_COO}.
Since the objective is to redistribute the entire cost associated to the members of $\coal$,  we start from the set of efficient vector allocations
\begin{equation}
\setV \triangleq \bigg\{p \in\mathbb{R}^{|\coal|} :\, \sum_{j\in\coal}p_j = v(\coal) \bigg\},
\label{eq_impdef}
\end{equation}
and define the \emph{excess} as the difference between the value the members of $\setS\subset\coal$ can achieve as standalone coalition, and the aggregate cost allocated over them by participating in $\coal$. 
\begin{defn}[Excess]
For any subcoalition $\setS\subseteq\coal$, the \emph{excess}  w.r.t. $p\in\setV$ is
\begin{equation*}
e(\setS,p) =  v(\setS) - \sum_{j\in\setS}{p_j},
\end{equation*}
and $e(\varnothing,p) = 0$.
\end{defn}
From~\eqref{eq_impdef} it follows $e(\coal,p) = 0$. This concept allows us to define the set of allocations for which no agent has an incentive to leave $\coal$ for joining a coalition $\setS\subset\coal$.
\begin{defn}[Core]
The \emph{core} of the TU game $\Gamma_{\coal}$ is the set
\begin{equation*}
\setO = \left\{p\in\setV :\, e(\setS,p)\leq 0,\,\forall \setS\subseteq\coal\right\}.
\end{equation*}
\end{defn}
It follows that all $p\in\setO$ fulfill individual rationality, i.e., $p_j\geq v(\{j\})$, for all $j\in\coal$, as well as group rationality, i.e., $\sum_{j\in\setS} p_j \geq v(\setS)$, for all $\setS\subseteq\coal$ (therefore $\setO\subseteq\setV$). This means that no $p\in\setO$ can be improved by a subcoalition $\setS\subset\coal$.
In contrast, for any $p\in\setV\setminus\setO$ there exists a set of players $\setS\subset\coal$ that can claim a better allocation through a \emph{demand} against $p$.
\begin{defn}[Demand]
A demand of a subcoalition $\setS\subset\coal$ against $p \in\setV$ is a pair $(\setS,\delta)$, where $\delta=e(\setS,p)$.
\end{defn}
It follows that a demand $(\setS,\delta)$ is satisfied by any allocation $p' = (p'_j)_{j\in\coal}$ such that $e(\setS,p')=0$.
This can be achieved by defining the new allocation $p'$ according to an egalitarian redistribution,
\begin{equation}\label{eq_satisfaction}
p'_j = \left\{
\begin{array}{ll}
	p_j - \frac{e(\setS,p)}{|\setS|}, & \text{if } j\in\setS,\\
	p_j + \frac{e(\setS,p)}{|\coal\setminus\setS|}, & \text{if } j\in\coal\setminus\setS.
\end{array}\right.
\end{equation}
Notice that $p'\in \setV$ since $\sum p'_j = \sum p_j = v(\coal)$.\par
Next, we provide the conditions under which there exists $p'\in\setV$ that can satisfy any demand $(\setS,\delta)$ from any $\setS\subset\coal$. In other words, the objective is to define some sufficient conditions for the nonemptiness of the core of a given subgame $\Gamma_{\coal}$. To do this, we will use the following assumption, whose implications are delineated in the next proposition.
\begin{assum}\label{assum_superadd}
Given a coalition $\coal\subseteq\setN$, there exists $\alpha\in\mathscr{L}$ such that $v(\coal) = \alpha(|\coal|) \sum_{i\in\coal}v(\{i\})$, and $\alpha(\cdot) \leq 1$.
\end{assum}
\begin{rem}\label{rem_superadd}
Note that assumption~\ref{assum_superadd} is an extension of~\eqref{eq_COO_OK} to $n\geq2$ players.
Let $J_{\setS}$ be the predicted cost for a coalition $\setS\in\coalstr(\mmg(\cdot))$ over the horizon $[k,k+N_p-1]$, as defined in Theorem~\ref{lem_lyap}, and $J_i\triangleq \sum_{t=0}^{N_p-1} \ell_i(x_i^{\ast}(t|k),u_i^{\ast}(t|k))$ the predicted selfish cost for agent $i\in\setN$ over the same horizon. Then, by~\eqref{eq_index_COO},~\eqref{eq_COO_OK}, and the convexity of the stage cost $\Lambda_i(\cdot,\cdot)$, Assumption~\ref{assum_superadd} holds for coalition $\setS\in\coalstr(\mmg(\cdot))$ if
\begin{equation*}
J_{\setS}\leq \alpha \sum_{i\in\setS} J_i - \Delta^{\chi}_{\setS}(\alpha)
\end{equation*}
holds for some $\alpha\leq 1$, where $\Delta^{\chi}_{\setS}(\alpha) \triangleq \chi_{\setS}-\alpha \sum_{i\in\setS}\chi_{i}>0$.
\end{rem}
\begin{lem}\label{lem_convex}
Let Assumption~\ref{assum_superadd} hold for all coalitions $\setS\subseteq\coal\in\coalstr$. Moreover, let $\alpha(\cdot)$ be such that
$\alpha(n)\leq 2\alpha(n-1) - \alpha(1)$. Then the core of the subgame $\Gamma_{\coal} = \langle\coal,v\rangle$ is nonempty.
\end{lem}
\begin{proof}
Convexity of a game implies nonemptiness of the core~\cite{Shapley1971Convex}. Given the definition of convex cost game~\cite{Driessen1988book}
\begin{equation*}
v(\setS\cup\{j\}) - v(\setS) \geq v(\setT\cup\{j\}) - v(\setT),\, \forall \setT\subseteq\setS\subseteq\coal\setminus\{j\},
\end{equation*}
for all $j\in\coal$, we use Assumption~\ref{assum_superadd} to obtain
\begin{multline*}
\alpha(|\setS|+1) \left(\sum_{i\in\setS} v(\{i\}) +v(\{j\}) \right) - \alpha(|\setS|) \sum_{i\in\setS} v(\{i\})  \\
\geq \alpha(|\setT|+1) \left(\sum_{i\in\setT} v(\{i\}) +v(\{j\}) \right) - \alpha(|\setT|) \sum_{i\in\setT} v(\{i\}),\\
\forall \setT\subseteq\setS\subseteq\coal\setminus\{j\}.
\end{multline*}
From here we derive the upper bound conditions on $\alpha(\cdot)$, over the possible coalition sizes, that guarantee the convexity of the game.
\end{proof}
Now we are ready to state the main property of the proposed algorithm, which stems from the work of~\cite{Cesco98aconvergent} and~\cite{Shapley1971Convex}. 
\begin{thm}\label{thm_convergTU}
Consider the subgame $\Gamma_{\coal}=\langle\coal,v\rangle$ associated to the coalition $\coal\in\coalstr$, and let Assumptions~\ref{assum_TU} and~\ref{assum_superadd} hold. Let $p^{(l)}$ be an allocation vector resulting from the reallocation mechanism described by~\eqref{eq_satisfaction} at a given iteration $l\in\mathbb{N}$. Let the core $\setO_{\coal}$ of $\Gamma_{\coal}$ be a nonempty set. Then the distance of any allocation to the core decreases at each successive iteration, 
i.e., $\min_{z\in\setO_{\coal}} \|p^{(l)}-z\|_2 \leq \min_{z\in\setO_{\coal}} \|p^{(l-1)}-z\|_2$. 
\end{thm}
\begin{proof}
The proof relies on~\cite{Cesco98aconvergent}, and on the nonemptiness of the core, established in Lemma~\ref{lem_convex}.
\end{proof}
\begin{rem}
This is not a sharp result. Convexity is a strong condition, sufficient but not necessary for the nonemptiness of the core. Indeed, the latter is directly connected with the less strict category of \emph{balanced} games. However, balancedness of a game needs to be numerically addressed even for a number of agents as low as five~\cite{Shapley1967Naval}. Hence, convergence of the redistribution mechanism may hold even if the convexity requirements established in Lemma~\ref{lem_convex} are not met.
\end{rem}
Finally, we report an additional relevant property of the algorithm.
\begin{coroll}\label{coroll_epscoreCesco}
Let $\setO = \varnothing$, and let $\setO(\varepsilon)$ be the \emph{least}-core, defined as
\begin{equation}
\setO(\varepsilon) = \left\{p\in\setV :\, e(\setS,p)\leq \varepsilon,\forall \setS\subseteq\coal\right\}.
\label{eq_epscoredef}
\end{equation}
where $\varepsilon \geq 0$ is the smallest such that $\setO(\varepsilon)$ is nonempty. Then the results of Theorem~\ref{thm_convergTU} apply to 
$z\in\setO(\varepsilon)$.
\end{coroll}
\begin{rem} 
Nonemptiness of the core can be checked in polynomial time if the complete description of the game is available~\cite{ChalkiadakisEtAl2012}. However, notice that one of the main features of the proposed algorithm is that the computation of the value of the complete subgame (i.e., $2^{|\coal|}$ possible pairs $\{\setS,\,\coal\setminus\setS\}$, with $\setS\subseteq\coal$, evaluated following the procedure in Section~\ref{sec_COO}) is not required for convergence. Under the informational constraints that characterize the system under study, this becomes very relevant from a practical point of view, as it substantially relaxes the cognitive and computational requirements of the proposed scheme~\cite{SANDHOLM1999}. 
\end{rem}

\section{Bargaining procedure}
\label{sec_bargaining}
At every time $k\in\setT_{\coal}$, Algorithm~\ref{alg_coalmpc} is executed. All players initiate a pairwise bargaining whose outcome will dictate the evolution of the coalitional structure. The procedure follows the evaluation of the coalitional benefit described in Section~\ref{sec_COO}. Note that the possible pairs are restricted to those considering dynamically coupled players.
\begin{rem}
In general it might not be viable to exhaustively evaluate all possible pairs of coalitions in $\{\coal_1,\ldots,\coal_{n_c}\}$. In practice, several (dynamically coupled) pairs $\setP_1,\setP_2\in\coalstr(\mmg(k))$ can be randomly selected; in this case the final outcome of the coalition formation process might be influenced by the random selection order~\cite{Jackson2010}.
\end{rem}
If condition~\eqref{eq_COO_OK} is verified for a given pair $\coal_1,\coal_2\in\coalstr$, the coalition $\coal_1\cup\coal_2$ is formed. The allocation of all the agents composing the new coalition is initialized by an equal share of the aggregate cost. Since this allocation is not necessarily stable (in the coalitional sense), Algorithm~\ref{alg_TUcoalsplit} is executed. 
Requests for utility transfer within $\coal_1\cup\coal_2$ are checked over a finite number of different subsets $\setS_{1}\subset\coal_1\cup\coal_2$ (note that the check is made over both $\setS_{1}$ and its complementary set $\setS_{2} = (\coal_1\cup\coal_2) \setminus \setS_{1}$).
If some subset of agents is dissatisfied with the currently assigned allocation, the iterative utility transfer scheme described in Section~\ref{sec_coalstab} is performed.\par
Algorithm~\ref{alg_TUcoalsplit} is performed also in the case where all the agents have already joined the grand coalition (no pairs are available), and in the case in which the merger between two players is not successful. 
Demands are checked similarly as described for the case in which a new merger is formed. Let $\coal_i$ be the coalition under analysis, and $\setS_1,\setS_2\subset\coal_i$. In this case the predictions for $v(\setS_1),v(\setS_2),v(\coal_i)$ will be updated according to the current state of the system, and condition~\eqref{eq_COO_OK} might not be fulfilled anymore. If this happens, either of the subsets $\setS_1,\setS_2$ will leave the coalition. Thus, while any coalition is formed through a bilateral agreement, a player can leave it unilaterally.
\begin{algorithm}
\caption{Bottom-up coalitional control}
\label{alg_coalmpc}
\begin{algorithmic}
\REQUIRE $\coalstr = \{\coal_1,\ldots,\coal_{n_c}\}$ 
\ENSURE coalition structure $\coalstr^+$, allocation vector $p^+\in\mathbb{R}^{|\setN|}$
\IF{$n_c> 1$}
\FORALL{pairs $\setP_1,\setP_2\in\coalstr(k)$ such that $\exists i\in\setP_1, j\in\setP_2:\, i\in\setM_j \, \vee \,j\in\setM_i$}
\STATE Call Algorithm~\ref{alg_coalform};
\ENDFOR
\ELSE
\STATE Call Algorithm~\ref{alg_TUcoalsplit} with $\setP_A = \setN$ and $\setP_B=\varnothing$.
\ENDIF
\end{algorithmic}
\end{algorithm}
\begin{algorithm}
\caption{Coalition formation}
\label{alg_coalform}
\begin{algorithmic}
\REQUIRE $\setP_1,\setP_2\in\coalstr$, {\ttfamily max\_iter}
\ENSURE $\coalstr^+$, allocation vector $(p_j^+)_{j\in\setP_1\unione\setP_2}$
\STATE $\vmerger \leftarrow$ minimize~\eqref{eq_coal_MPC} over $\bfu_{1\unionped 2}$;
\STATE initialize $\tilde{\bfu}_{1}$ and $\tilde{\bfu}_{2}$ with feasible trajectories; 
\FOR{$t=1,\ldots,$ {\ttfamily max\_iter}}
	\STATE $v(\setS_1) \leftarrow$ solve~\eqref{eq_coal_MPC} over $\bfu_1$ w.r.t. $\tilde{\bfu}_{2}$, replacing~\eqref{eq_mod_MPC} with~\eqref{eq_mpc_unilateral};
	\STATE $v(\setS_2) \leftarrow$ solve~\eqref{eq_coal_MPC} over $\bfu_2$ w.r.t. $\tilde{\bfu}_{1}$, replacing~\eqref{eq_mod_MPC} with~\eqref{eq_mpc_unilateral};
	\STATE $\tilde{\bfu}_{1}\leftarrow\bfu_{1}^{\ast}$; $\quad\tilde{\bfu}_{2}\leftarrow\bfu_{2}^{\ast}$;
\ENDFOR
\IF{\eqref{eq_COO_OK} is verified}
\STATE $\coalstr^+\leftarrow$ $\coalstr\setminus\{\setP_1,\setP_2\}\cup\{\setP_1\cup\setP_2\}$; \COMMENT{Form coalition} 
\STATE $p_j := \vmerger/|\setP_1\cup\setP_2|$ for all $j\in\setP_1\cup\setP_2$; \COMMENT{Initialize with egalitarian allocation}
\STATE $(p^+_j)_{j\in\setP_1\unione\setP_2}\leftarrow$ Call Algorithm~\ref{alg_TUcoalsplit};
\ELSE
\STATE Call Algorithm~\ref{alg_TUcoalsplit} for $\setP_1$ and $\setP_2$. \COMMENT{Check for demands within each coalition}
\ENDIF
\end{algorithmic}
\end{algorithm}
\begin{algorithm}[t]
\caption{Utility transfer / Coalition splitting}
\label{alg_TUcoalsplit}
\begin{algorithmic}[1]
\REQUIRE $\setP\in\coalstr$, allocation vector $(p_j)_{j\in\setP}$, {\ttfamily max\_loops}, {\ttfamily max\_iter} 
\ENSURE $\{\coal'_1,\ldots,\coal'_{n_c}\}$ where $\bigcup_i \coal'_i = \setP$, allocation vector $(p_j')_{j\in\bigcup_i \coal'_i}$
\STATE $n_c\leftarrow 0$;
\STATE {\ttfamily n\_loops}$\leftarrow 0$;
\STATE {\ttfamily flag}$\leftarrow$FALSE;
	\REPEAT
		\STATE (randomly) choose $\setS_1\subset\setP$; 
		\STATE $\setS_2\leftarrow\setP\setminus\setS_1$; 	
			\STATE $\bfp_1\leftarrow \sum_{j\in\setS_1}p_j$; $\quad\bfp_2\leftarrow \sum_{j\in\setS_2}p_j$;
		\STATE initialize $\tilde{\bfu}_{1}$ and $\tilde{\bfu}_{2}$ with feasible trajectories; 
		\FOR{$t=1,\ldots,$ {\ttfamily max\_iter}}
			\STATE $v(\setS_1) \leftarrow$ solve~\eqref{eq_coal_MPC} over $\bfu_1$ w.r.t. $\tilde{\bfu}_{2}$, replacing~\eqref{eq_mod_MPC} with~\eqref{eq_mpc_unilateral};
			\STATE $v(\setS_2) \leftarrow$ solve~\eqref{eq_coal_MPC} over $\bfu_2$ w.r.t. $\tilde{\bfu}_{1}$, replacing~\eqref{eq_mod_MPC} with~\eqref{eq_mpc_unilateral};
			\STATE $\tilde{\bfu}_{1}\leftarrow\bfu_{1}^{\ast}$; $\quad\tilde{\bfu}_{2}\leftarrow\bfu_{2}^{\ast}$;
		\ENDFOR
		\IF {$\bfp_i +\bfp_j > v(\setS_i)+v(\setS_j)$ } 
			\STATE {\ttfamily flag}$\leftarrow$TRUE; \COMMENT{$\setP_r$ splits into $\setS_1$ and $\setS_2$}
			\FOR{$i=1,2$}
			\STATE $j\in\{1,2\}\setminus\{i\}$;
				\STATE $n_c\leftarrow n_c + 1$;
				\STATE $\coal_{n_c}\leftarrow\setS_i$;
				\STATE Initialize payoff of every agent in $\coal_{n_s}$ by equally splitting $v(\setS_i)$;
				\STATE Call Algorithm~\ref{alg_TUcoalsplit} for $\coal_{n_s}$;
			\ENDFOR
		\ELSE 
			\IF {$\bfp_1>v(\setS_1)$ \textbf{or} $\bfp_2>v(\setS_2)$} 
				\STATE $e\leftarrow \bfp_1 - v(\setS_1)$;
				\STATE $p_j'\leftarrow p_j - e/|\setS_1|$ for $j\in\setS_1$; \COMMENT{Satisfy demand}
				\STATE $p_j'\leftarrow p_j + e/|\setS_2|$ for $j\in\setS_2$;
			\ENDIF
		\ENDIF
		\STATE {\ttfamily n\_loops}$\leftarrow ${\ttfamily n\_loops}$+1$;
	\UNTIL{{\ttfamily flag}=FALSE} \textbf{and} {\ttfamily n\_loops}$<${\ttfamily max\_loops}.
\end{algorithmic}
\end{algorithm}
\begin{rem}
The procedure is independent for every pair of players, and the execution of the algorithm can be parallelized.
\end{rem}
Finally, for all instants $k\in\mathbb{N}$, the allocation of each agent $j\in\coal_i$ is updated as
\begin{equation}
p_j(k) = \frac{p_j(k')}{v(\coal_i)} \left[\Lambda_i(\bfx_i(k),\bfu_i(k)) + \chi_i(\mmg_i(k'))\right],
\label{eq_updatealloc}
\end{equation}
where $k'\in\setT_{\coal}$ is the time corresponding to the last execution of Algorithm~\ref{alg_coalmpc}, and $v(\coal_i)$ is the value of coalition $\coal_i$ computed at time $k'$.

\section{Example}
\label{sec_example}
To test the proposed algorithm, we address the wide-area control (WAC) of power networks~\cite{ChakraborttyKhargonekar2013ACC}. The objective of WAC is to damp inter-area oscillations arising among connected generators, causing undesired power transfers. These oscillations have been poorly controllable with local (decentralized) control. The development of flexible AC transmission systems (FACTS) and the recent availability of a capillary network of sensors such as the phasor measurement units (PMUs) have opened new possibilities in the control of the power grid~\cite{RiversoEtAl2013}. Yet the dense information exchange between PMUs installed at substations managed by different utility companies is not free of costs, and research has been focusing on the development of WAC strategies promoting the sparsity of inter-area communications~\cite{LianEtAl2017_JSAC,DorflerEtAl2014}.
\subsection{System description}
The power network consists of several areas coupled by transmission lines (see Fig.~\ref{fig_schemaPNS}). Local generation is available within each area. 
Our focus in on the load-frequency control (LFC) loop, to \emph{(i)} maintain the frequency around the nominal value, and \emph{(ii)} reduce power transfers between areas. In particular, we test the proposed framework in providing automatic generation control (AGC) to regulate the frequency to its nominal value in presence of step changes in the load.\par
\begin{figure}
	\centering
	\def\svgwidth{0.8\columnwidth}
	\scriptsize{
		\input{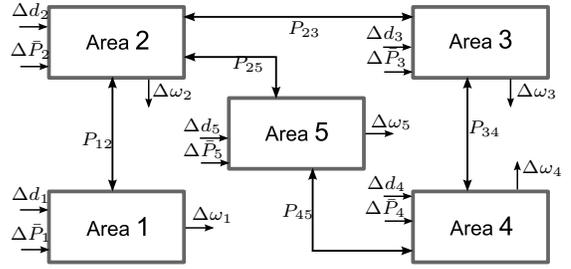}
		}
	\caption{Power network composed of 5 areas with local supply~\cite{RiversoPSNBenchmark2012}. Power transfers are possible between areas connected by transmission lines. The objective is to control inter-area oscillations---cause of undesired power transfers---and to minimize the deviation from the nominal frequency under step variations in the load. Two cases are considered: \emph{(i)} local production capacity is sufficient for locally matching the demand, \emph{(ii)} the capacity of local generation is impaired, making energy transfers from neighboring areas necessary for demand satisfaction.}
	\label{fig_schemaPNS}
\end{figure}
\begin{table}[tbp]
\renewcommand{\arraystretch}{1.3}
\caption{Symbols employed in the power network example.}
\label{tab_ex_symbol}
\centering
   \begin{tabular}{ c p{0.68\columnwidth} c }
	Symbol & Description & Unit \\ \hline
     $\Delta d$ & Deviation of the load from the nominal value (p.u.) & [-] \\ 
     $\Delta\theta$ & Variation in the rotor angle w.r.t. revolving magnetic field & [rad] \\
		 $\Delta\omega$ & Deviation from the nominal frequency & [rad/s] \\
		 $\Delta P_{m}$ & Deviation from the nominal mechanical power (p.u.) & [-] \\
		 $\Delta P_{v}$ & Deviation from the nominal steam valve position (p.u.) & [-] \\
		 $\Delta\bar{P}$ & Deviation of the power setpoint from the nominal value (p.u.) & [-] \\
		 $H$ & Machine inertia constant & [s] \\ 
		 $r_v$ & Rotor velocity regulation  & [rad/s] \\ 
		 $\rho_f$ & Load change / frequency variation (\%) & [-] \\
		 $\tau_{t}$ & Prime mover time constant & [s] \\
		 $\tau_{g}$ & Governor time constant & [s] \\ 
		 $P^{0}_{ij}$ & Synchronizing power coefficient & [rad$^{-1}$] %
   \end{tabular}
\end{table}
The energy supply in each area is provided by a power station equipped with single-stage turbines. We consider resistive loads, not sensitive to frequency variations (e.g., lighting, heating). Let each area be identified by an index in the set $\setN = \{1,\ldots,5\}$. The linearized dynamics of synchronous generators in area $i\in\setN$ result in the following continuous-time model~\cite[Chap.~12]{SaadatPowerSysAnalysisBook}:
\begin{equation}
\dot{x}_i = A_{ii}x_i + B_{ii} u_i + D_i \Delta d_{i} + \sum_{j\in\setM} A_{ij}x_j,
\label{eq_PNScontmodel}
\end{equation}
where $x_i\triangleq[\Delta\theta_i,\Delta\omega_i,\Delta P_{m_i},\Delta P_{v_i}]\in\mathbb{R}^{4}$, $u_i = \Delta\bar{P}_i\in\mathbb{R}$, and $\Delta d_i\in\mathbb{R}$ is the variation in the demand (symbols are defined in Table~\ref{tab_ex_symbol}).
The last term in~\eqref{eq_PNScontmodel} describes the influence of coupled areas, identified in the set $\setM_i = \{j\in\setN\setminus\{i\} |\, A_{ij}\neq \mathbf{0}\}$. Matrices are composed as
{\arraycolsep=1.4pt\renewcommand{\arraystretch}{1.1}
\begin{equation}\label{eq_PNScontmat}
\begin{split}
A_{ii} & = \left[
\begin{array}{cccc}
0 & 1 & 0 & 0\\
-\frac{\sum_{j\in\setM_i}P^{0}_{ij}}{2H_i} & -\frac{\rho_{f_i}}{2H_i} & \frac{1}{2H_i} & 0\\
0 & 0 & -\frac{1}{\tau_{t_i}} & \frac{1}{\tau_{t_i}}\\
0 & -\frac{1}{r_{v_i} \tau_{g_i}} & 0 & -\frac{1}{\tau_{g_i}}
\end{array}\right]
\quad
B_i = \left[
\begin{array}{c}
0\\
0\\
0\\
\frac{1}{\tau_{g_i}}
\end{array}\right] \\
A_{ij} & = \left[
\begin{array}{cccc}
0 & 0 & 0 & 0\\
\frac{P^{0}_{ij}}{2H_i} & 0 & 0 & 0\\
0 & 0 & 0 & 0\\
0 & 0 & 0 & 0
\end{array}\right]
\quad
D_i = \left[
\begin{array}{c}
0\\
-\frac{1}{2H_i}\\
0\\
0
\end{array}\right].
\end{split}
\end{equation}}
For reasons of space, the values of the parameters are not reported here (the reader is referred to~\cite{RiversoPSNBenchmark2012}). The coupling of the generation frequency between areas connected through transmission lines appears in the second row of $A_{ii}$ and $A_{ij}$. For small deviations from the nominal value, inter-area power flows can be modeled as~\cite{SaadatPowerSysAnalysisBook}
\begin{equation}
\Delta P_{ij} = P^{0}_{ij} (\Delta \theta_i - \Delta \theta_j),\, i,j\in\setN,
\label{eq_powtransf}
\end{equation}
where $P^{0}_{ij}$, referred to as the \emph{synchronizing coefficient}, is the slope of the power-angle curve at the initial operating angle $\Delta \theta_{ij_0} = \Delta \theta_{i_0} - \Delta \theta_{j_0}$ between areas $i$ and $j$. These flows appear as a load increase in one area, and a load decrease in the other area. In particular, positive values of $\Delta P_{ij}$ indicate a transfer from area $i$ to area $j$.\par
Classic discretization yields non-sparse structures, unless very small sampling steps are employed~\cite{VadigepalliDoyle2003}. In order to preserve the topology of the system in the structure of the discrete-time model while avoiding the dependence on the sampling time, the continuous-time model~\eqref{eq_PNScontmodel} is discretized following the method of~\cite{FarinaEtAl2013}, with $T_s = 1\,\mathrm{s}$. More specifically, by treating $u_i$ as an exogenous input along with $\Delta d_{i}$ and $x_j$, the input-decoupled structure of the continuous-time model is replicated in discrete time. Notice that the use of such a method is reasonable in this kind of framework, where one basic assumption is that system-wide knowledge of the model is not likely to be achieved (besides communication constraints, one further reason is the dependence of the time constants characterizing the linear model on the current setpoints~\cite{DorflerEtAl2014}). From now on, any mention of the above matrices will refer to the discrete-time model.
\subsection{Controller design}
It can be inferred from~\eqref{eq_powtransf} that large energy transfers are caused by large differences in the angle deviation. The minimization of the energy transferred between connected areas can be implicitly addressed by penalizing large values of $\Delta \theta_i$; additionally, measures available from cooperating nodes can be exploited by penalizing the angle difference between the members of a given coalition. Therefore, the state weighting matrices in the objective function are chosen as
\begin{equation}\label{eq_weighttransfer}
\begin{split}
Q_{ii} & = \mathrm{diag}(q_{ii}^{\theta}+\sum_{j\in\setM_i}{q_{ij}^{\theta}+q_{ji}^{\theta}},q_{ii}^{\omega},q_{ii}^{P_{m}},q_{ii}^{P_{v}}),\\
Q_{ij} & = \mathrm{diag}(-\sum_{j\in\setM_i}q_{ij}^{\theta}+q_{ji}^{\theta},0,0,0),
\end{split}
\end{equation}
where $Q_{ij}\in\mathbb{R}^{n_i\times n_j}$ is the submatrix of $Q\in\mathbb{R}^{n\times n}$ relative to the coupling between nodes $i$ and $j$.
For noncooperative control, $q_{ij}=q_{ji}=0$; the rest of the values are defined as in~\cite{RiversoPSNBenchmark2012}, i.e., $Q_{ii} = \mathrm{diag}(500,0.01,0.01,10)$ and $R_i = 10$, $\forall i\in\setN$. In case of cooperation, we set $q_{ij}=q_{ji}=1000$.\par
\begin{figure}
\centering
\subfloat{
\label{fig_deltaPlim_1e_3}
\includegraphics[trim={0 0 0 0},clip,width=0.48\columnwidth]{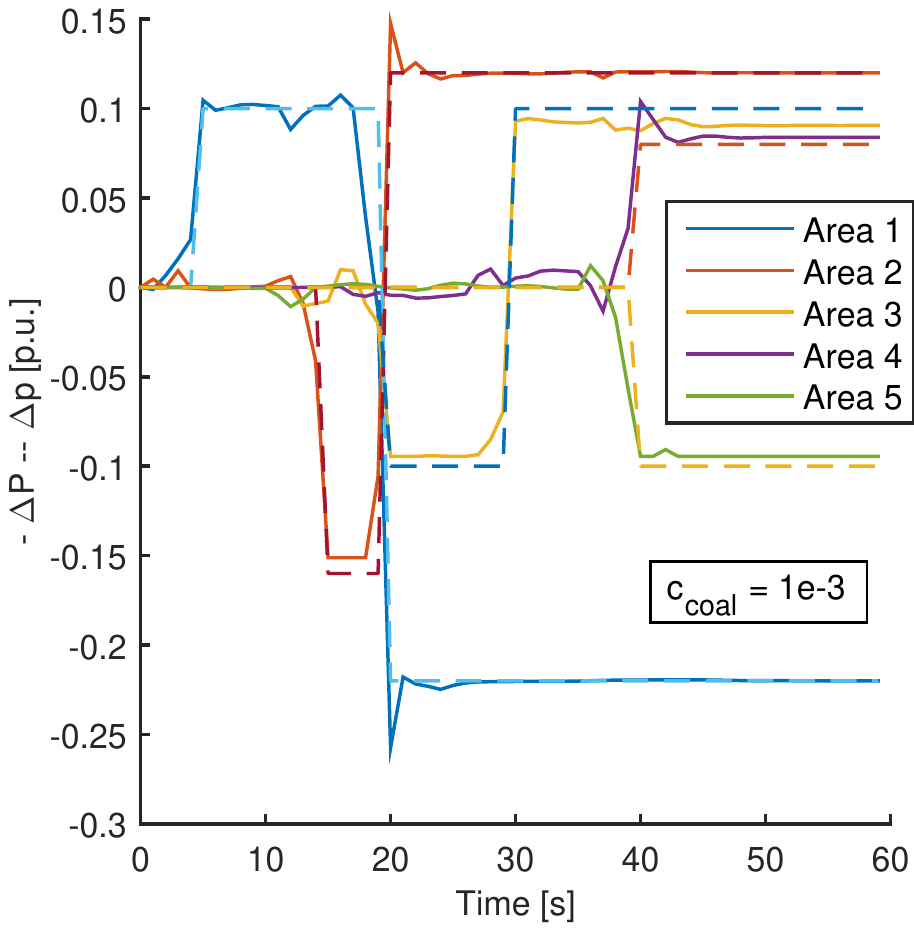}
}%
\subfloat{
\label{fig_deltaPlim_5e_4}
\includegraphics[trim={0 0 0 0},clip,width=0.48\columnwidth]{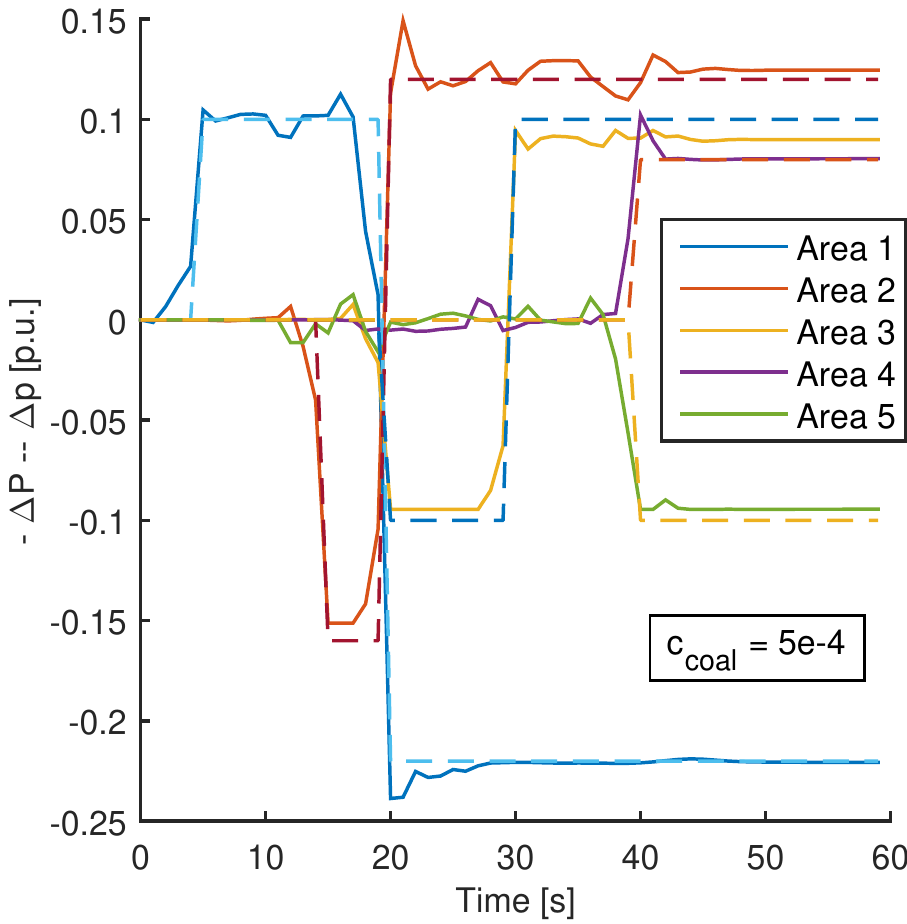}
}%
\caption{Scenario 2: the capacity of local generation is impaired, making energy transfers from neighboring areas necessary for demand satisfaction. Variation of the demand (dashed lines) and local power generation in the 5 areas. In the left plot, relative to $c_{\mathrm{coal}}=10^{-3}$, the lack of supply in Area~3 due to a 10\% capacity drop w.r.t the demand value is supplemented with energy transfers from Areas 4 and 5. Similarly, the right plot, corresponding to $c_{\mathrm{coal}}=5\cdot 10^{-4}$, shows that Area 3 receives additional supply from Areas 2 and 5 (see Fig.~\ref{fig_coalevoS2}). Power setpoints are computed with the RTO~\eqref{eq_RTO_ex}.}
\label{fig_deltaPlim}
\end{figure}
\begin{table}
\renewcommand{\arraystretch}{1.3}
\caption{Constraints on local generation.}
\label{tab_ex_constr_gen}
\centering
   \begin{tabular}{ c c c c c c }
	& $\left\|u_1\right\|_{\infty}\leq$ & $\left\|u_2\right\|_{\infty}\leq$ & $\left\|u_3\right\|_{\infty}\leq$ & $\left\|u_4\right\|_{\infty}\leq$ & $\left\|u_5\right\|_{\infty}\leq$\\ \hline
    S1 &  0.2310 & 0.1680 & 0.1050 & 0.0840 & 0.1050 \\ 
		S2 & 0.3465 & 0.1512 & 0.0945 & 0.1260 & 0.0945
   \end{tabular}
\end{table}
We test the capability of the coalitional controller based on autonomous coalition formation in achieving $\Delta \omega_i\rightarrow 0$ for all $i\in\setN$ in presence of step variations in the load $\Delta d_i$. Two scenarios are considered: in the first, local production capacity is sufficient for locally matching any demand, and the objective is to track the AGC reference $(\bar{x},\bar{u})$, computed as a function of the change in the grid load. Since each area's load must be matched with the local production, the components of the setpoint vector are defined as $\bar{x}_i = (0,0,\Delta d_i,\Delta d_i)$, $\bar{u}_i = \Delta d_i$, corresponding to the increment in the energy generation required to balance an increase in the demand. In the second scenario the capacity of local generation is impaired, making energy transfers from neighboring areas necessary for demand satisfaction (see Table~\ref{tab_ex_constr_gen}). These transfers are described by the coalitional setpoints optimized by an RTO layer
\begin{subequations}
\label{eq_RTO_ex}
\begin{align}
\min_{\bfu_i^{\mathrm{ref}},\bfx_i^{\mathrm{ref}}}  & \sum_{t=0}^{N_p-1} \left\|\bfu_i^{\mathrm{ref}}(t|k) - \bar{\bfu}_i(t|k)\right\|_{\bfR_i}^2\nonumber\\[-1.5ex]
& \qquad \qquad + \left\|\bfx_i^{\mathrm{ref}}(t+1|k) - \bar{\bfx}_i(t+1|k)\right\|_{\bfQ_i}^2  \label{eq_RTO_cost} \\ 
\mathrm{s.t.}& \nonumber\\
& \bfx_i(t+1|k)(I-\bfA_{ii}) - \bfB_{ii}\bfu_i(t|k)  = \mathbf{D}_i \delta_{i}(t|k),\label{eq_RTO_steadystate}\\
& \mathbf{1}\trasp \bfu_i(t|k)  = \mathbf{1}\trasp \delta_i(t|k),\label{eq_RTO_demsuplmatch}\\
& \bfu_i(t|k) \in\prod_{j\in\coal_i}\setU_j,\,t = 0,\ldots,N_p-1,\label{eq_restru_RTO}\\
& \sum_r P_{rj}(\Delta\theta_r - \Delta\theta_j) = [\Delta d_j - u_j^{\max}]_{+},\nonumber\\[-1.5ex]
& \qquad \qquad \qquad \qquad \qquad j\in\coal_i,\,\forall r\in\setM_j\cap\coal_i,\label{eq_trasf_RTO}\\
& \delta(t|k) = \hat{\delta}_i(k+t),\,t = 0,\ldots,N_p-1,\label{eq_RTO_disturb_estim}
\end{align}
\end{subequations}
where $\bfu_i^{\mathrm{ref}}\triangleq (\bfu_i^{\mathrm{ref}}(k),\ldots,\bfu_i^{\mathrm{ref}}(k+N_p-1))$ is the input reference trajectory along the horizon $N_p$ for $\coal_i$, and $\bfx_i^{\mathrm{ref}}\triangleq (\bfx_i^{\mathrm{ref}}(k+1),\ldots,\bfx_i^{\mathrm{ref}}(k+N_p))$ is the associated state reference. In the steady-state condition~\eqref{eq_RTO_steadystate}, $\delta_i \triangleq (\Delta d_j)_{j\in\coal_i}$ is the demand vector relative to all members of the coalition; \eqref{eq_RTO_demsuplmatch} defines the demand-supply equilibrium within a coalition, i.e., $\sum_{j\in\coal_i} \Delta \bar{P}_j = \sum_{j\in\coal_i} \Delta d_j$. In~\eqref{eq_trasf_RTO}, $[\cdot]_{+} \triangleq \max(\cdot,0)$. 
The quadratic coalitional stage cost in~\eqref{eq_RTO_cost} is defined by the weighting matrices $\bfQ_i = \mathrm{diag}(Q_{ii})$ and $\bfR_i = \mathrm{diag}(R_{ii})$, with $Q_{ii}=(10, 0, 100, 100)$ and $R_{ii} = 100$.
The setpoint $(\bfu_i^{\mathrm{ref}},\bfx_i^{\mathrm{ref}})\equiv (\bar{x}_i,\bar{u}_i)$ is assigned to singleton coalitions, since power transfers cannot be arranged for them.\par
The procedure described in Section~\ref{sec_bargaining} is followed to evaluate the possible formation of coalitions. At each time step the MPC problem~\eqref{eq_coal_MPC}---reformulated accordingly for the tracking of references $(\bfx_i^{\mathrm{ref}},\bfu_i^{\mathrm{ref}})$---is independently solved by the coalitions in $\coalstr(\mmg(k))$~\cite{Antoniovo2013Aut}.
Cooperation costs are defined as $\chi = c_{\mathrm{coal}}|\coal|^2$, for $|\coal|\geq 2$, $\chi = 0$ otherwise. The prediction horizon length is set to $N_p = 5$. Following~\cite{LimonEtAl2006_TAC} and setting $\bfQ_i^{\mathrm{f}} = 20 \bfQ_i$, the terminal cost in~\eqref{eq_coal_MPC_cost} is 
\begin{equation}
V_i^{\mathrm{f}}(\bfx_i(N_p|k)) = (\bfx_{i}(k) -\bar{\bfx}_{i} )\trasp \bfQ_i^{\mathrm{f}} (\bfx_{i}(k) -\bar{\bfx}_{i} ).
\label{eq_term_cost_powergrid}
\end{equation}
\subsection{Results}
\begin{figure}
	\centering
	\subfloat{
\label{fig_1e_3_S1}
\includegraphics[trim={0 0 0 0},clip,width=0.95\columnwidth]{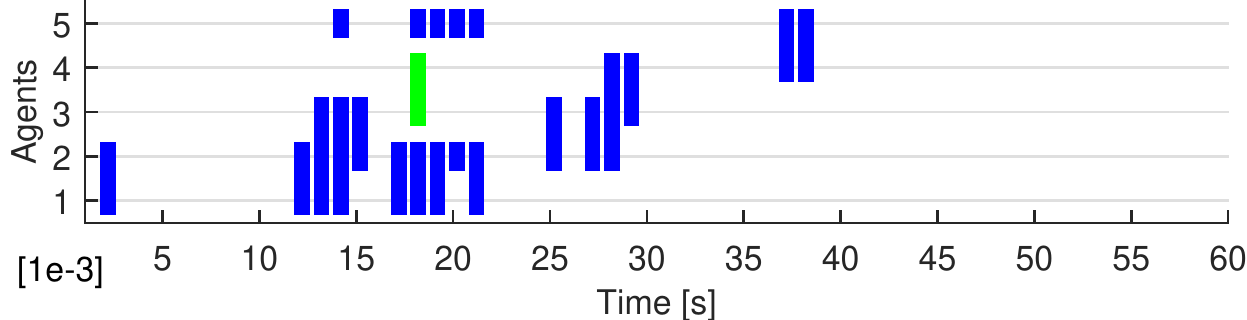}
}\\
\subfloat{
\label{fig_5e_4_S1}
\includegraphics[trim={0 0 0 0},clip,width=0.95\columnwidth]{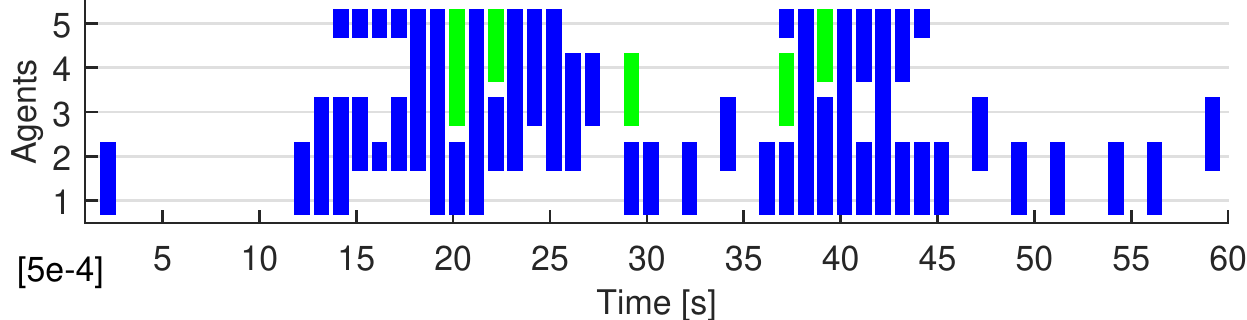}
}\\
\subfloat{
\label{fig_1e_4_S1}
\includegraphics[trim={0 0 0 0},clip,width=0.95\columnwidth]{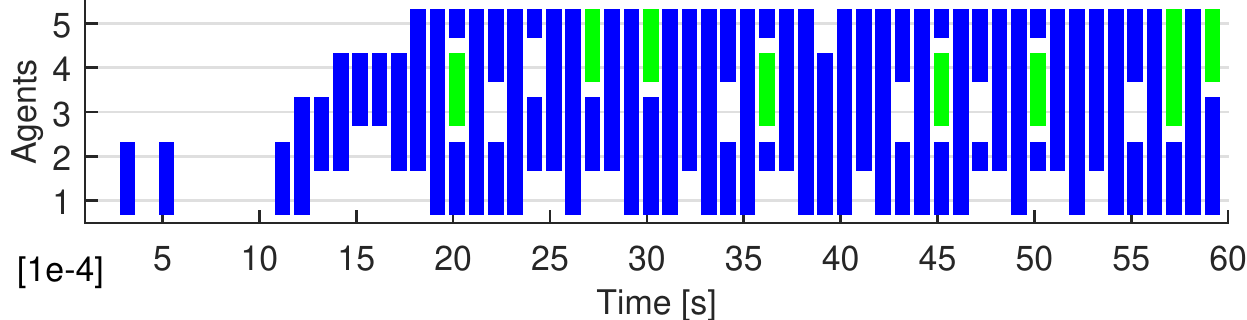}
}\\
	\subfloat{
\label{fig_1e_5_S1}
\includegraphics[trim={0 0 0 0},clip,width=0.95\columnwidth]{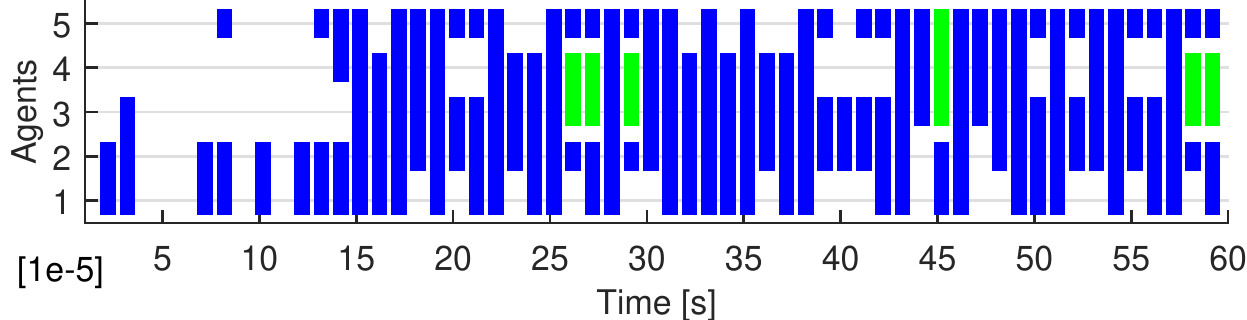}
}%
\caption{Scenario 1: production capacity is sufficient for locally matching any demand. Formation of coalitions for different values of $c_{\mathrm{coal}}$. Costs of cooperation are increasing with the coalition size, i.e., $\chi = c_{\mathrm{coal}}|\coal|^2$, for $|\coal|\geq 2$.}
	\label{fig_coalevoS1}
\end{figure}
\begin{figure}
	\centering
	\subfloat{
\label{fig_1e_3_S2}
\includegraphics[trim={0 0 0 0},clip,width=0.95\columnwidth]{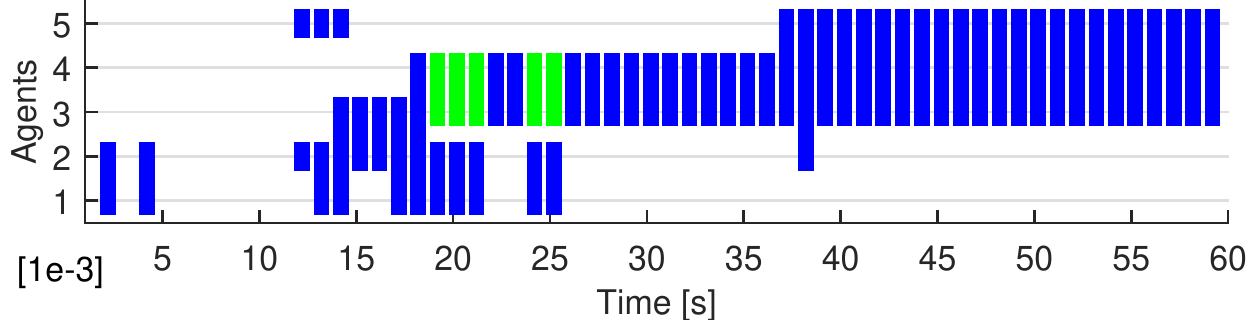}
}\\
\subfloat{
\label{fig_5e_4_S2}
\includegraphics[trim={0 0 0 0},clip,width=0.95\columnwidth]{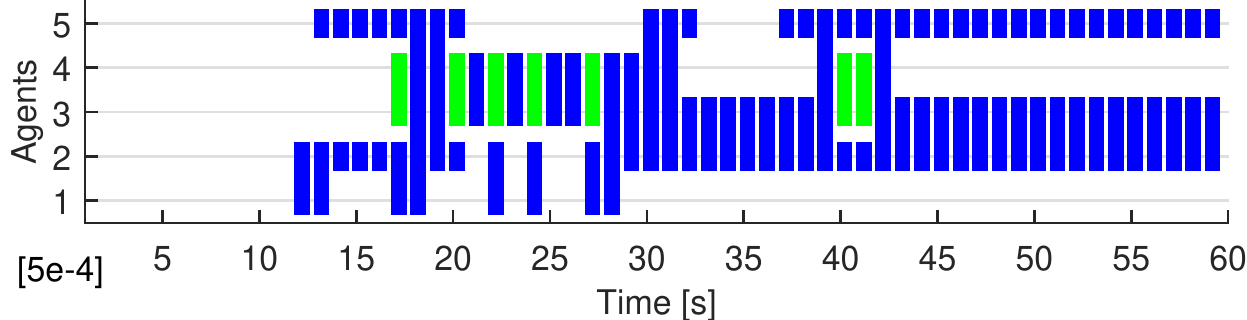}
}\\
\subfloat{
\label{fig_1e_4_S2}
\includegraphics[trim={0 0 0 0},clip,width=0.95\columnwidth]{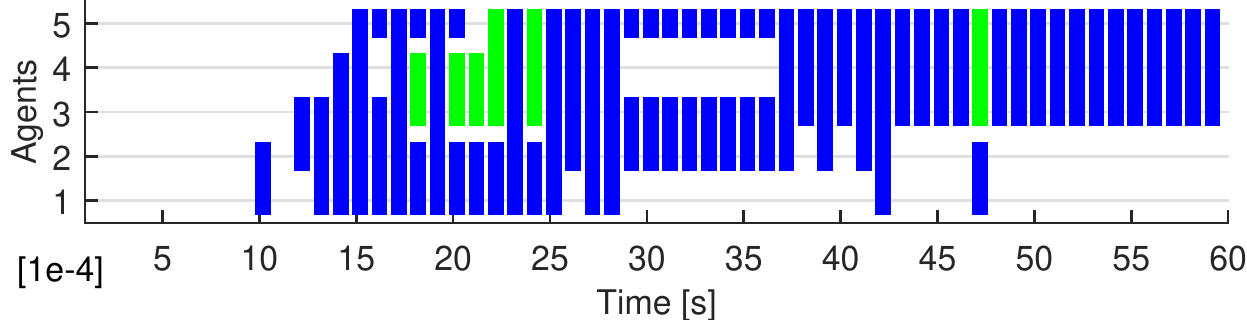}
}\\
	\subfloat{
\label{fig_1e_5_S2}
\includegraphics[trim={0 0 0 0},clip,width=0.95\columnwidth]{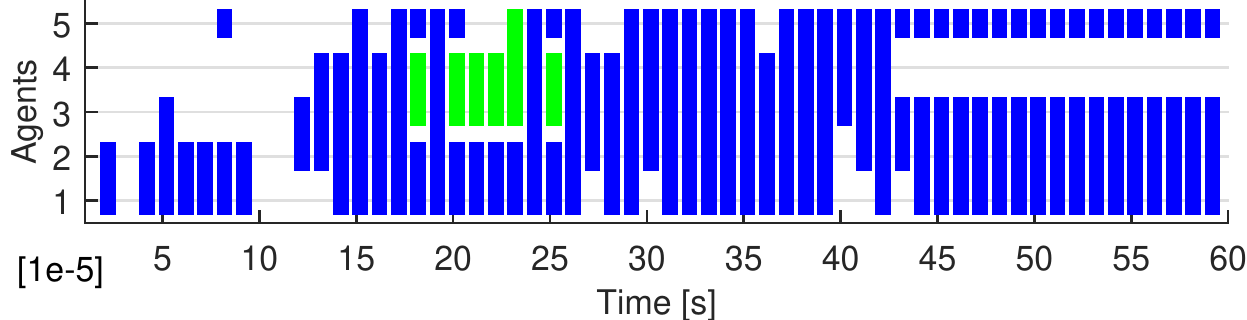}
}%
\caption{Scenario 2: the capacity of local generation is impaired, making energy transfers from neighboring areas necessary for demand satisfaction (see Fig.~\ref{fig_deltaPlim}). The plots show the evolution of coalitions for different values of $c_{\mathrm{coal}}$. Costs of cooperation increase with the coalition size, i.e., $\chi = c_{\mathrm{coal}}|\coal|^2$, for $|\coal|\geq 2$. The cooperation in this case follows a more stable pattern.}
	\label{fig_coalevoS2}
\end{figure}
In order to evaluate the variation of the controller performance over different degrees of cooperation, two indices are defined. 
The first is the average overall frequency deviation,
\begin{equation}
\eta(\omega) = \frac{1}{T_{\mathrm{sim}}}\sum_{t=1}^{T_{\mathrm{sim}}}\sum_{i\in\setN} \Delta\omega_i^2,
\label{eq_eta_f}
\end{equation}
and the second reflects the energy transferred between areas,
\begin{equation}
\psi(\theta) = \sum_{t=1}^{T_{\mathrm{sim}}}\sum_{i\in\setN}\sum_{j\in\setM_i} \left\|\Delta P_{ij}(t)T_s\right\|^2,
\label{eq_psi2}
\end{equation}
where $\Delta P_{ij}$ is defined in~\eqref{eq_powtransf}, and $T_s$ is the sampling time. These indices provide a measure of the global performance not dependent of the particular evolution of the coalition structure.\par 
In this case, the supply capacity in Area 3 is not always sufficient to fulfill the local demand; meanwhile, generators in Area~5 cannot decrease their production to match the lowest local demand level, so the excess of production is transferred to other areas. As can be seen in Figure~\ref{fig_deltaPlim}, the lack of supply capacity is covered by neighboring generators.\par
\begin{figure}
\centering
\subfloat{
\label{fig_eta}
\includegraphics[trim={0 0 0 0},clip,height=\altfigur]{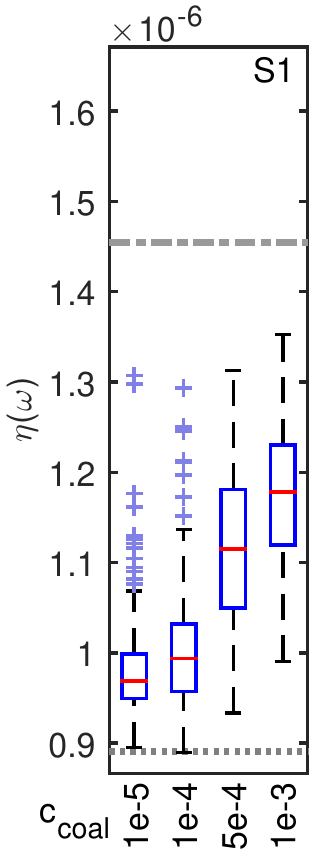}
}\hspace*{-0.7em}%
\subfloat{
\label{fig_eta_lim}
\includegraphics[trim={0 0 0 0},clip,height=\altfigur]{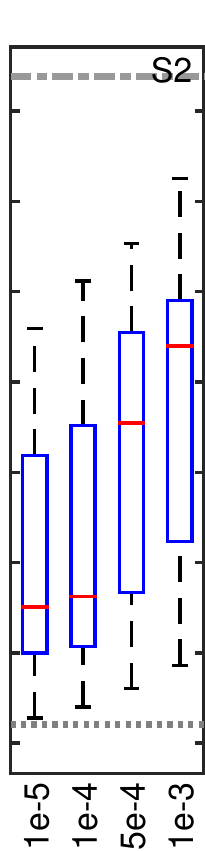}
}\hspace*{3em}%
\subfloat{
\label{fig_psi}
\includegraphics[trim={0 0 0 0},clip,height=\altfigur]{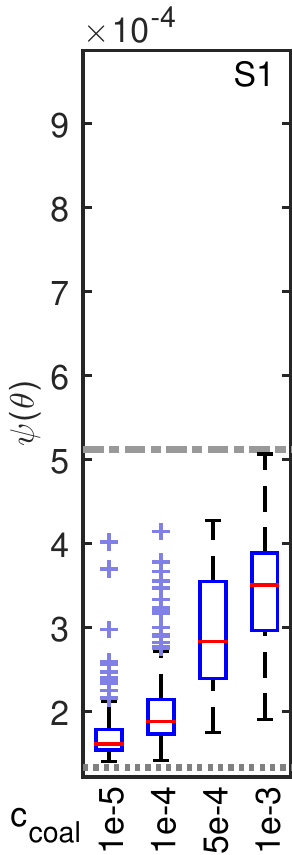}
}\hspace*{-0.7em}%
\subfloat{
\label{fig_psi_lim}
\includegraphics[trim={0 0 0 0},clip,height=\altfigur]{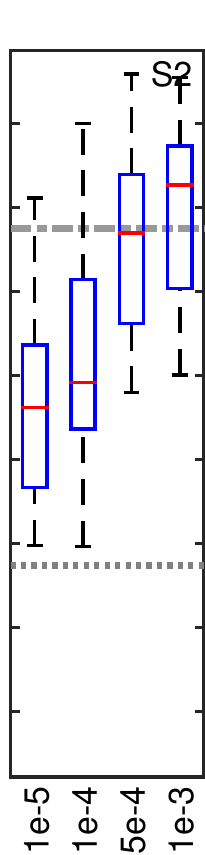}
}%
\caption{Performance index $\eta(\omega)$ (left) and $\psi(\theta)$ (right), respectively regarding the minimization of the frequency deviation and of inter-area energy transfers, for increasing values of $c_{\mathrm{coal}}$. 
Plots marked with S1 are relative to Scenario~1, while S2 refers to the case in which areas 2, 3 and 5 experience limitations in their power generation. The dotted line marks the performance of the strictly cooperative strategy (centralized MPC), whereas the dashed-dotted line refers to the strictly noncooperative one. See Fig.~\ref{fig_TU_S1} for details on the box representation. Even with scarce cooperation ($c_{\mathrm{coal}}=10^{-3}$), the performance improvement over noncooperative control is sensible: indices $\eta(\omega)$ and $\psi(\theta)$ are enhanced in Scenario~1 by about 18\% and 31\%, respectively. In Scenario~2, $\eta(\omega)$ is improved by about 18\%; however, power transfers cannot be avoided in this scenario, and the low coordination between areas results in an increase of $\psi(\theta)$ by 5\%.}
\label{fig_etaspsis}
\end{figure}
\begin{figure}
\centering
\subfloat{
\label{fig_coalsiz}
\includegraphics[trim={0 0 0 0},clip,height=\altfigur]{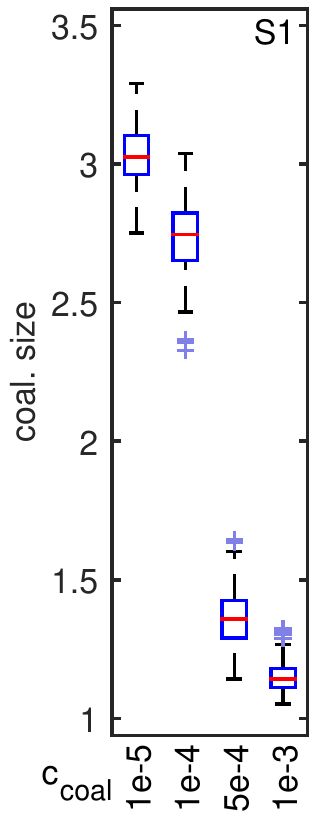}
}\hspace*{-0.7em}%
\subfloat{
\label{fig_coalsiz_lim}
\includegraphics[trim={0 0 0 0},clip,height=\altfigur]{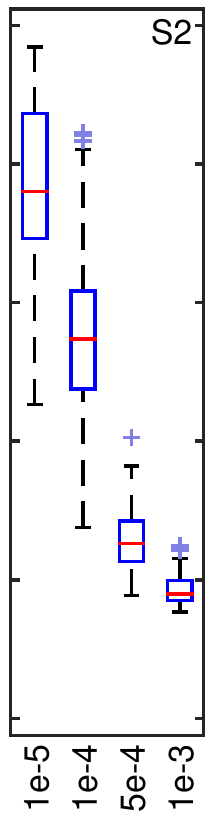}
}\hspace*{3em}
\subfloat{
\label{fig_coaldur}
\includegraphics[trim={0 0 0 0},clip,height=\altfigur]{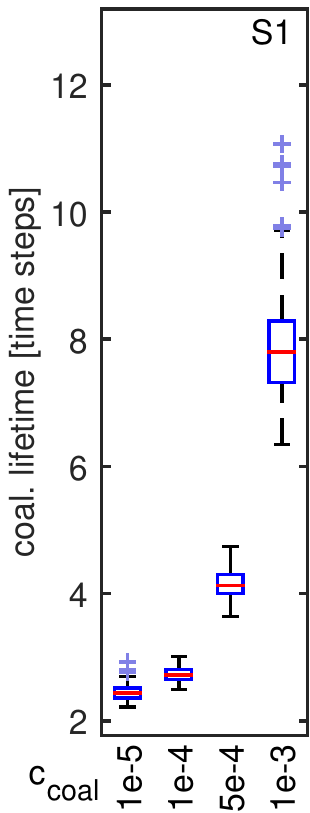}
}\hspace*{-0.7em}%
\subfloat{
\label{fig_coaldur_lim}
\includegraphics[trim={0 0 0 0},clip,height=\altfigur]{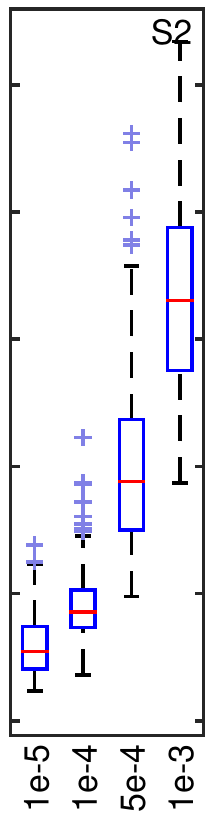}
}%
\caption{Average size of coalitions (left) and coalition lifetimes (right), for different values of $c_{\mathrm{coal}}$ (costs of cooperation are increasing with the coalition size, i.e., $\chi = c_{\mathrm{coal}}|\coal|^2$). Plots marked with S1 are relative to Scenario 1, while S2 refers to the case in which areas 2, 3 and 5 experience limits on the power generation. See Fig.~\ref{fig_TU_S1} for details on the box representation.}
\label{fig_coalsizdurs}
\end{figure}
Figures~\ref{fig_etaspsis} and~\ref{fig_coalsizdurs} gather the results of a set of 200 simulations for the two scenarios, showing the performance for different values of $c_{\mathrm{coal}}$. Coalition formation is disincentivized as $c_{\mathrm{coal}}$ is increased, deteriorating the achievable performance. Roughly speaking, the performances of coalitional control fall between those obtained through fully-cooperative (centralized) and noncooperative MPC control. It is interesting to see how, even with a reduced cooperation effort, the performance improvement over the noncooperative control is sensible: with $c_{\mathrm{coal}}=10^{-3}$ (see top plot in Fig.~\ref{fig_coalevoS1}), yielding an average coalition size of 1.2, indices $\eta(\omega)$ and $\psi(\theta)$ are enhanced by about 18\% and 31\%, respectively. In Scenario 2, $\eta(\omega)$ is improved by about 18\%; however, power transfers cannot be avoided in this scenario, and the low coordination between areas results in an increase of $\psi(\theta)$ by 5\%.\par
Table~\ref{tab_TUalloc} shows the accumulated control costs for each area, in Scenario~1 (cooperation costs are not included). The allocation produced by the proposed iterative utility transfer algorithm is compared to the Shapley value. In order to better evaluate these two outputs, the agents were not allowed to leave the grand coalition in the simulations relative to Table~\ref{tab_TUalloc}. 
The first two columns show the control costs associated to centralized (fully cooperative) and noncooperative MPC: notice how for Area~1 cooperation implies an increase of the local cost. Individual rationality is achieved for all areas with both allocation methods. The results relative to the iterative transfer algorithm have been obtained with 10 iterations, i.e., the dissatisfaction w.r.t. the assigned allocation has been checked for 10 randomly selected subcoalitions (see Section~\ref{sec_bargaining}). Instead, the Shapley value required at each time step the evaluation of all possible subcoalitions, in this case $2^{5} = 32$.\par
Figures~\ref{fig_TU_S1} and~\ref{fig_TU_S2} show the accumulated control costs for the 5 areas, and their corresponding online reallocation, resulting over 200 simulations for the two scenarios. Notice how---particularly in Scenario~1---individual rationality is not always fulfilled when cooperation costs become appreciable. Online reallocation mitigates this issue and provides an incentive for the cooperation (see especially the case of Area~1). 
\begin{figure*}
\centering
\subfloat{
\label{fig_loc_cost1}
\includegraphics[trim={0 0 0 0},clip,height=\altfigur]{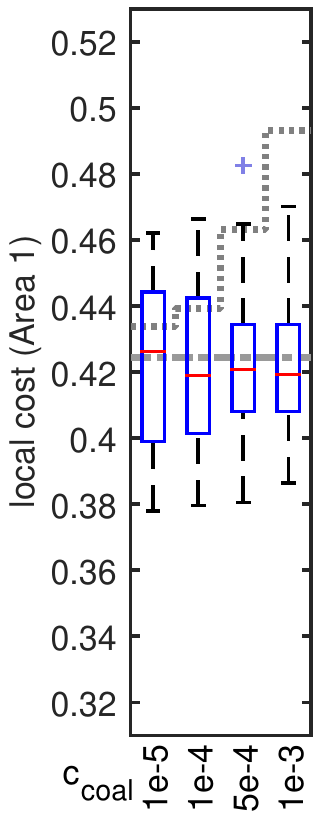}
}\hspace*{-0.7em}%
\subfloat{
\label{fig_loc_costTU1}
\includegraphics[trim={0 0 0 0},clip,height=\altfigur]{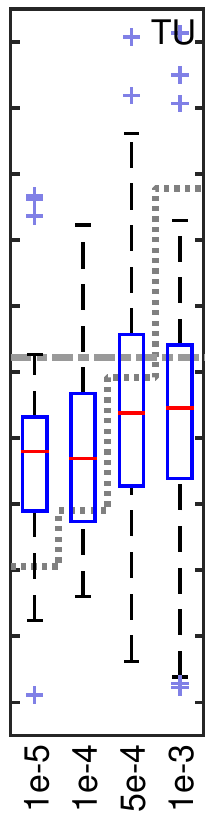}
}\hspace*{-0.3em}
\subfloat{
\label{fig_loc_cost2}
\includegraphics[trim={0 0 0 0},clip,height=\altfigur]{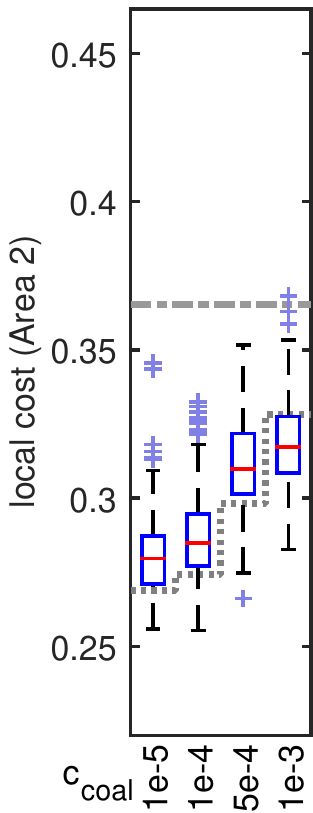}
}\hspace*{-0.7em}%
\subfloat{
\label{fig_loc_costTU2}
\includegraphics[trim={0 0 0 0},clip,height=\altfigur]{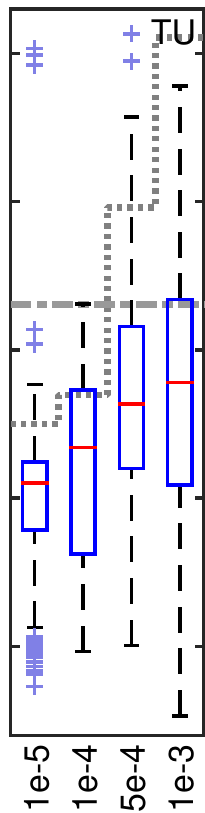}
}\hspace*{-0.3em}
\subfloat{
\label{fig_loc_cost3}
\includegraphics[trim={0 0 0 0},clip,height=\altfigur]{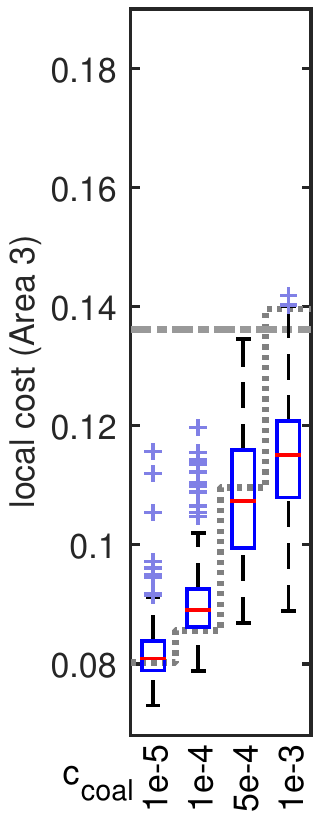}
}\hspace*{-0.7em}%
\subfloat{
\label{fig_loc_costTU3}
\includegraphics[trim={0 0 0 0},clip,height=\altfigur]{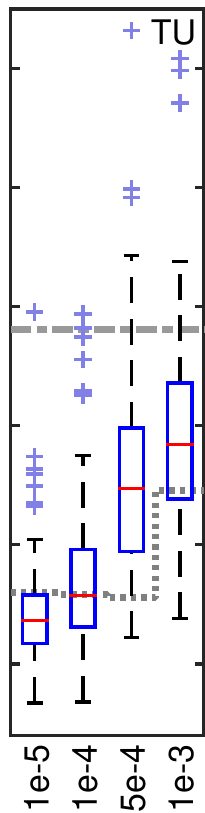}
}\hspace*{-0.3em}
\subfloat{
\label{fig_loc_cost4}
\includegraphics[trim={0 0 0 0},clip,height=\altfigur]{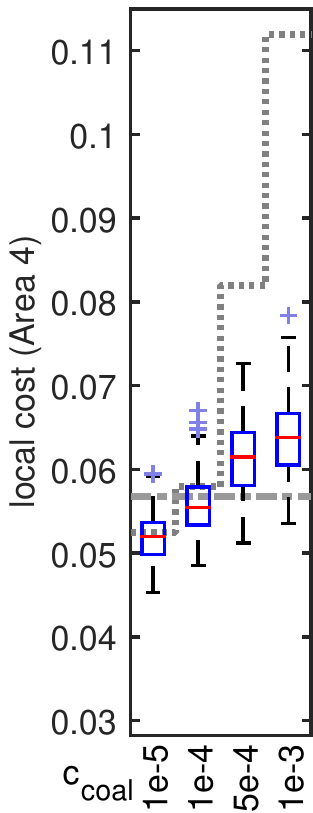}
}\hspace*{-0.7em}%
\subfloat{
\label{fig_loc_costTU4}
\includegraphics[trim={0 0 0 0},clip,height=\altfigur]{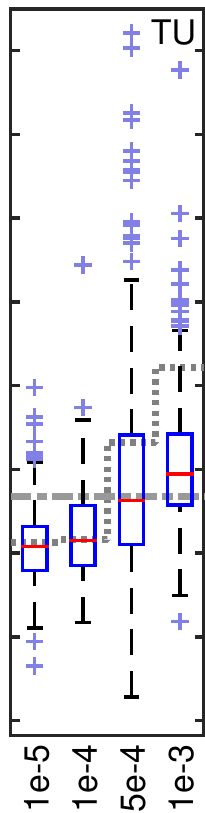}
}\hspace*{-0.3em}
\subfloat{
\label{fig_loc_cost5}
\includegraphics[trim={0 0 0 0},clip,height=\altfigur]{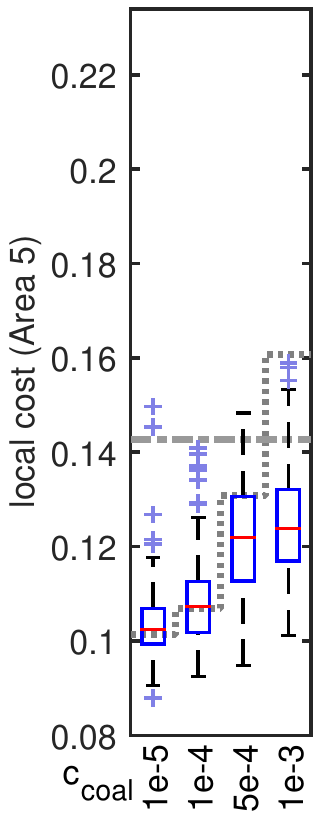}
}\hspace*{-0.7em}%
\subfloat{
\label{fig_loc_costTU5}
\includegraphics[trim={0 0 0 0},clip,height=\altfigur]{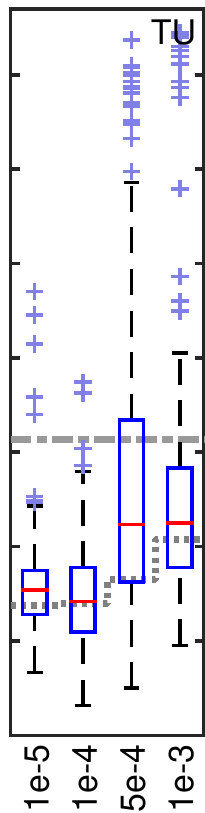}
}%
\caption{Scenario 1: production capacity is sufficient for locally matching any demand. Cost---involving both control and cooperation---locally incurred by agents, accumulated along the simulated interval. Plots marked with TU show the result of the online reallocation with the proposed algorithm. Box plots gather the results of 200 simulations, for different values of $c_{\mathrm{coal}}$ (costs of cooperation are increasing with the coalition size, i.e., $\chi = c_{\mathrm{coal}}|\coal|^2$, for $|\coal|\geq 2$). 
As a reference, the costs corresponding to the fully cooperative strategy are denoted by the dotted line, showing the influence of cooperation costs. These---initially equally supported by the agents---are reallocated online with the proposed algorithm, as shown by the dotted line in `TU' plots. The dashed-dotted line refers to the local cost with the noncooperative strategy. Boxes cover the range between the 25th and the 75th percentiles (the central mark is the median), and outliers (data exceeding a distance from the box extremes of 1.5 times the difference between the 25th and the 75th percentiles) are plotted separately.}
\label{fig_TU_S1}
\subfloat{
\label{fig_loc_cost1_S2}
\includegraphics[trim={0 0 0 0},clip,height=\altfigur]{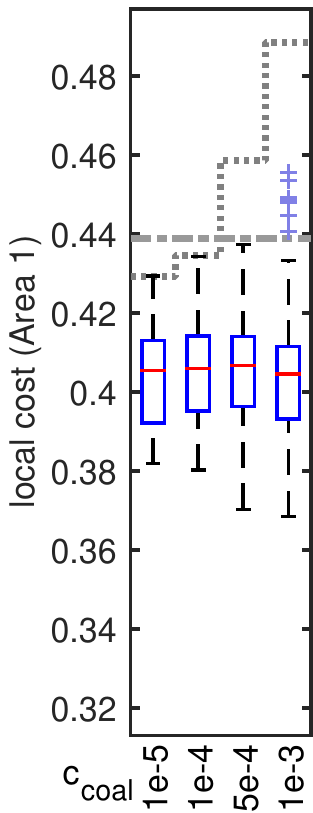}
}\hspace*{-0.7em}%
\subfloat{
\label{fig_loc_costTU1_S2}
\includegraphics[trim={0 0 0 0},clip,height=\altfigur]{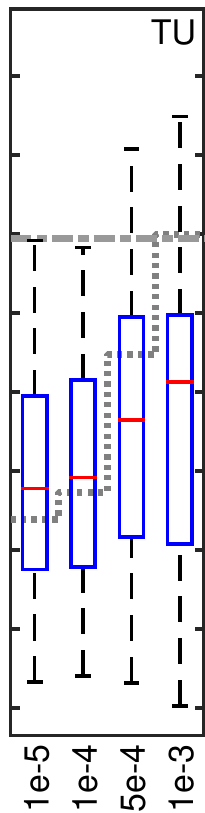}
}\hspace*{-0.3em}
\subfloat{
\label{fig_loc_cost2_S2}
\includegraphics[trim={0 0 0 0},clip,height=\altfigur]{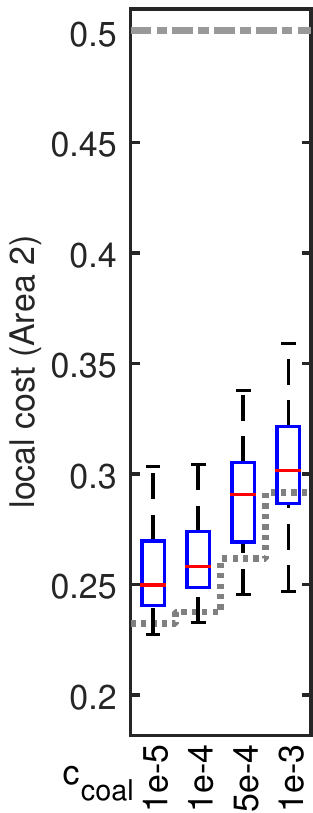}
}\hspace*{-0.7em}%
\subfloat{
\label{fig_loc_costTU2_S2}
\includegraphics[trim={0 0 0 0},clip,height=\altfigur]{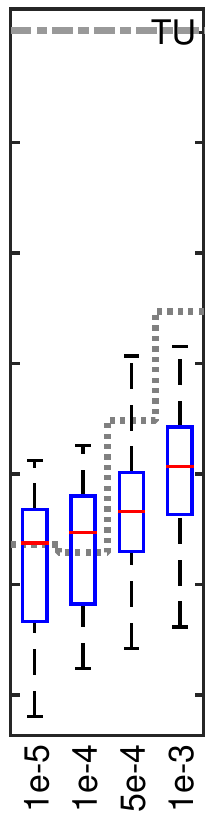}
}\hspace*{-0.3em}
\subfloat{
\label{fig_loc_cost3_S2}
\includegraphics[trim={0 0 0 0},clip,height=\altfigur]{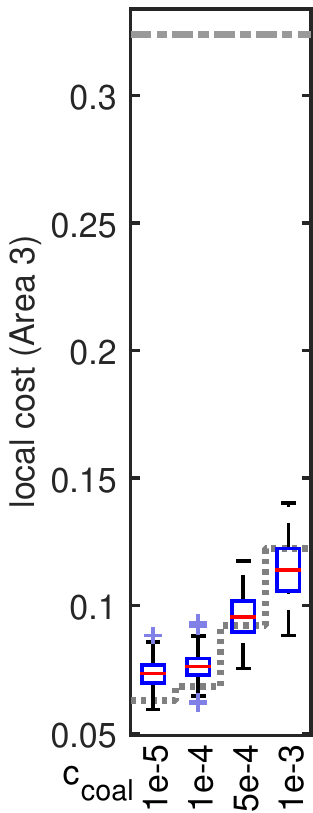}
}\hspace*{-0.7em}%
\subfloat{
\label{fig_loc_costTU3_S2}
\includegraphics[trim={0 0 0 0},clip,height=\altfigur]{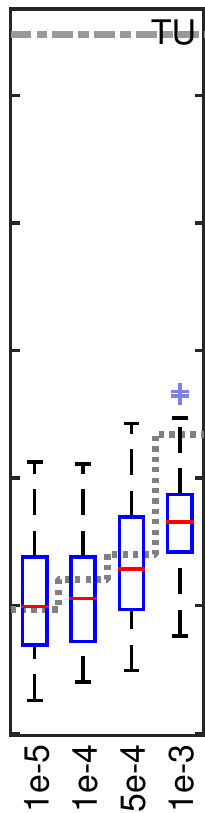}
}\hspace*{-0.3em}
\subfloat{
\label{fig_loc_cost4_S2}
\includegraphics[trim={0 0 0 0},clip,height=\altfigur]{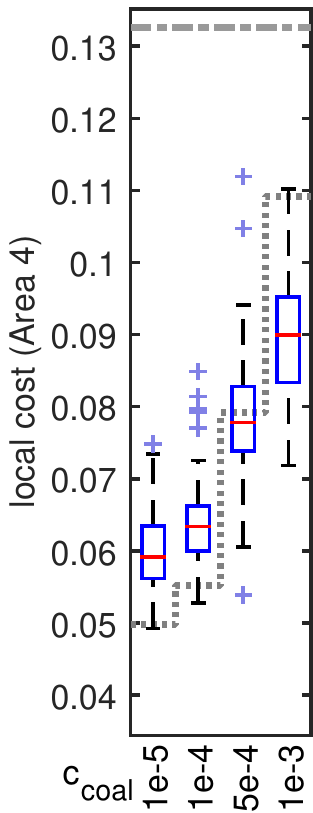}
}\hspace*{-0.7em}%
\subfloat{
\label{fig_loc_costTU4_S2}
\includegraphics[trim={0 0 0 0},clip,height=\altfigur]{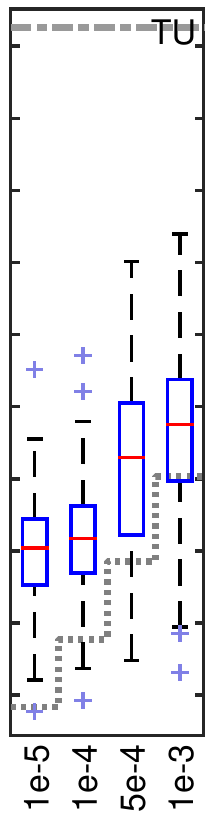}
}\hspace*{-0.3em}
\subfloat{
\label{fig_loc_cost5_S2}
\includegraphics[trim={0 0 0 0},clip,height=\altfigur]{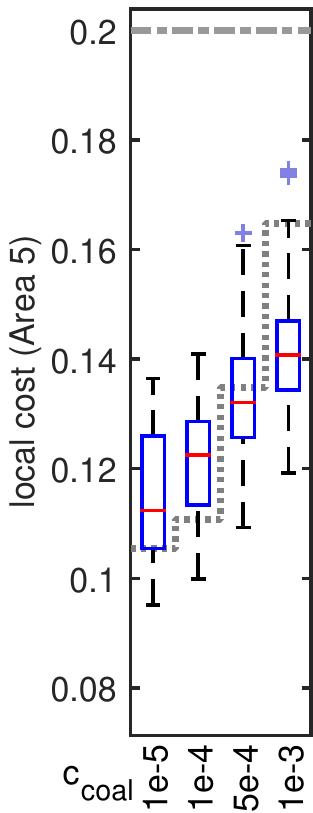}
}\hspace*{-0.7em}%
\subfloat{
\label{fig_loc_costTU5_S2}
\includegraphics[trim={0 0 0 0},clip,height=\altfigur]{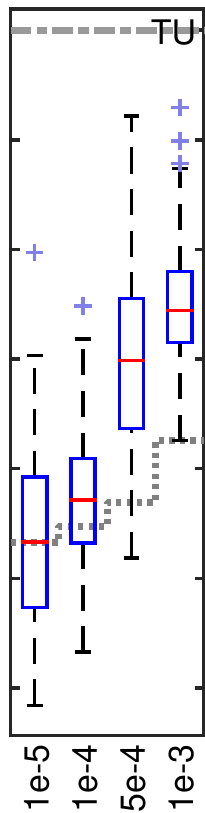}
}%
\caption{Scenario 2: the capacity of local generation is impaired, making energy transfers from neighboring areas necessary for demand satisfaction. Control cost locally incurred by agents, accumulated along the simulated interval. Plots marked with TU show the result of the online cost reallocation with the proposed algorithm. As a reference, the costs corresponding to the fully cooperative strategy are denoted by the dotted line, showing the influence of cooperation costs. These---initially equally supported by the agents---are reallocated online with the proposed algorithm, as shown by the dotted line in `TU' plots. The dashed-dotted line refers to the local cost with the noncooperative strategy. See Fig.~\ref{fig_TU_S1} for details on the box representation.}
\label{fig_TU_S2}
\end{figure*}
\begin{table}
\caption{TU scheme allocation for the grand coalition.}
\label{tab_TUalloc}
\centering
   \begin{tabular}{ c c  c c  c }
	 	& Dec. MPC & Centr. MPC & TU alg. & Shapley \\
		\hline%
    Area 1 & 0.424 & 0.433 & 0.359 & 0.353 \\
    Area 2 & 0.365 & 0.268 & 0.329 & 0.333 \\
		Area 3 & 0.136 & 0.080 & 0.085 & 0.085 \\
		Area 4 & 0.057 & 0.052 & 0.054 & 0.052 \\
		Area 5 & 0.143 & 0.101 & 0.110 & 0.112 
   \end{tabular}
\end{table}

\section{Conclusion and outlook}
\label{sec_concl}
A coalitional control framework based on a switching model predictive control (MPC) architecture for large-scale systems is proposed in this paper. In particular, the framework is directed at systems with heterogeneous character, in which the autonomous components possibly  pursue competing objectives, and possess a limited model of the system. 
By characterizing as a transferable-utility cooperative game the benefit provided to local control agents by a broader feedback and the pursuit of a common objective, the formation of coalitions of controllers is promoted accordingly.
The redistribution of the coalitional benefit is used as incentive for cooperation. Taking into account the informational constraints of the considered setting, a proper allocation of the control cost is derived without the need of a complete knowledge of the game.
The analysis shows that, when global model information is only partially available, cooperation costs play a major role on the outcome of the coalition formation, and that these can be used as a mechanism to link coalition formation with desired closed-loop properties.
Simulation results from a case study of power grid wide-area control show that the reconfiguration capabilities provided to the system through the proposed framework are suited for fault-tolerance needs or plug-and-play settings.\par 
One of the most interesting control challenges arising in the considered setting comes from an informational point of view. The effect of circumscribed information availability on the overall system stability have recently been subject of study~\cite{TanakaEtAl2017ACC,MylvaganamAstolfi2016ACC,DerooEtAl2015}.
Depending on the system dynamics and on the performance requirements, a matter of study could be the design of the terminal ingredients employed in the receding-horizon optimization. A non-conservative design of these elements generally requires information only available at global level, although the actual synthesis can be distributed across the agents~\cite{ConteEtAl2016aut}. We believe this is an interesting topic when privacy concerns need to be taken into account.
Another aspect to be further addressed is the deviation of the actual realization of the cooperation benefit from the expected one, on which the coalitional agreeement and the allocation mechanism are based~\cite{BaeyensEtAl2013}.
Future work might also consider overlapping coalitions, as a means for enhancing the flexibility of cooperation and providing further possibilities for the dynamic reallocation of the agents' control effort.
\section*{Acknowledgment}
The authors would like to thank the anonymous referees and the editor, whose recommendations substantially contributed to improve this work. The authors are extremely grateful to Dr. Antonio Ferramosca and Dr. David Angeli for the enlightening discussions, and to Prof. Daniel Lim\'on and Prof. Teodoro \'Alamo for their constant support. 

\bibliographystyle{ieeetr}
\bibliography{biblio_tesi}

\end{document}